\documentclass[11pt]{article}
\usepackage[utf8]{inputenc}
\usepackage[normalem]{ulem}
\usepackage{mathtools}
\usepackage{bbm}
\usepackage{graphicx}
\usepackage{longtable}
\usepackage{wrapfig}
\usepackage{rotating}
\usepackage[normalem]{ulem}
\usepackage{amsmath}
\usepackage{amssymb}
\usepackage{capt-of}
\usepackage{hyperref}
\usepackage{epsfig,psfrag,latexsym}
\usepackage{amscd}
\usepackage{color}
\usepackage{amsfonts}
\usepackage{wasysym}
\usepackage{mathrsfs}
\usepackage{verbatim}
\newcounter{thm}
\newtheorem{theorem}[thm]{Theorem}

\newtheorem{prop}[thm]{Proposition}
\newtheorem{corollary}[thm]{Corollary}

\newtheorem{lemma}[thm]{Lemma}

\newenvironment{proof}[1][Proof]{\textbf{#1.} }{\ \rule{0.5em}{0.5em}}

\renewcommand{\mathbbm}[1]{\boldsymbol{#1}}
\author{Margherita Disertori\footnote{Institut f{\"u}r Angewandte Mathematik \&
Hausdorff Center for Mathematics,Endenicher Allee 60, 53115 Bonn, Germany. E-mail:disertori@iam.uni-bonn.de},  Valentin Rapenne\footnote{Institut Camille Jordan,Univ. Lyon 1, 43 bd. du 11 nov.
1918, 69622 Villeurbanne cedex,France. E-mail:valentin.rapenne.maths@gmail.com}\\
 Constanza Rojas-Molina\footnote{Laboratoire AGM, D\'epartement de Math\'ematiques,CY Cergy Paris Universit\'e, 2 Av. Adolphe Chauvin, 95302 Cergy-Pontoise, France. E-mail:crojasmo@cyu.fr}, Xiaolin Zeng\footnote{IRMA, Universit\'e de Strasbourg,
7 rue Ren\'e Descartes, 67084 Strasbourg, France. E-mail:zeng@math.unistra.fr }}
\date{}
\title{Phase transition in the Integrated Density of States of the Anderson model arising
from a supersymmetric sigma model}
\hypersetup{
 pdfauthor={Margherita Disertori, Constanza Rojas-Molina, Valentin Rapenne, Xiaolin Zeng},
 pdftitle={Density of state of the random Schr\"odinger operator arising from supersymmetric hyperbolic sigma model},
 pdfkeywords={},
 pdfsubject={},
 pdfcreator={Emacs 28.2 (Org mode 9.5.5)}, 
 pdflang={English}}
\begin{document}

\maketitle
\begin{abstract}
We study the Integrated Density of States (IDS) of the random Schr\"odinger operator appearing in the study of certain reinforced random processes in connection with a supersymmetric sigma-model. 
We rely on previous results on the supersymmetric sigma-model to obtain lower and upper bounds  on  the asymptotic behavior of the IDS near the bottom of the spectrum in all dimension. 
We show a phase transition for the IDS between weak and strong disorder regime in dimension larger or equal to three, that follows from a phase transition  in the corresponding random process and supersymmetric sigma-model. 
In particular,  we show that the IDS does not exhibit Lifshitz tails in the strong disorder regime, confirming a recent conjecture. This is in stark contrast with other disordered systems, like the Anderson model. A Wegner type estimate is also derived, giving an upper bound on the IDS and showing the regularity of the function. 
\\[3mm]
2020 MSC: 82B44, 82B20 (primary), 47B80, 60K37, 60G60  (secondary).\\
Keywords: random Sch\"odinger operator, nonlinear hyperbolic supersymmetric sigma model,  density of states, processes with reinforcement.
\end{abstract}

\section{Introduction and main results}
\label{sec:org992929a}

Transport phenomena in disordered materials can be described at the quantum mechanical level via
random Sch\"odinger operators. On the lattice \(\mathbb{Z}^{d},\) \(d\geq 1,\) they generally take the
form of an infinite random matrix
\(H_{\omega }= -\Delta+\lambda V_{\omega }\in
\mathbb{R}^{\mathbb{Z}^{d}\times \mathbb{Z}^{d}}_{sym}\) where \(-\Delta \) is the negative discrete Laplacian 
and \(V_{\omega }\) is a diagonal matrix with random entries.

In this paper we consider a random Schr\"odinger operator \(H_{\beta }\) (defined in \eqref{eq:defHbetazd} below)
arising from the supersymmetric hyperbolic sigma model \(H^{2|2}\) introduced by Zirnbauer in the context of quantum diffusion
\cite{Zirnbauer1991}  \cite{drunk1992migdal}. This can be seen as a statistical mechanics spin model, where the
spins take values on a supersymmetric extension of the hyperbolic plane \(H^{2}.\) 
This model  is expected to qualitatively reflect the phenomenon of Anderson localization and delocalization for real
symmetric band matrices (see  \cite[Section 3]{Disertori2010a}) and  exhibits a dimension-dependent phase transition
between a disordered phase \cite{Disertori2010} and an ordered phase with spontaneous symmetry breaking
\cite{Disertori2010a}.

In recent years, \(H^{2|2}\) attracted a growing interest from the mathematics
community due to the discovery in  \cite{Sabot2015,Sabot2017, Sabot2019} of surprising connections with
two linearly  reinforced random processes: the edge reinforced random walk introduced by  Diaconis in 1986,
and the vertex reinforced jump process conceived by Werner around 2000.
A first spectacular application  of this connection  was the proof of a phase transition in the reinforced processes
between a recurrent and a transient phase that follows from the disorder/order transition in the \(H^{2|2}\) model \cite{Angel2014lo,Sabot2015,Disertori2015tr}.

In \cite{Sabot2017,Sabot2019}, Sabot, Tarr\` es and Zeng show that a key ingredient in proving many properties of the reinforced processes  is the connection with the random Schr\"odinger operator \(H_\beta\) \eqref{eq:defHbetazd},
which is the main object of study in the present paper.  The spectrum of this random operator is deterministic (see \cite[Theorem 3.10]{AW} and \cite[Chap. 4]{theserapenne}) and the existence/non-existence of an eigenvalue in \(0\)
is related to  transience/recurrence properties of the stochastic processes \cite{Sabot2017, Sabot2019}.
At large disorder  the spectrum is pure point  \cite{Collevecchio2018}.

In this paper we pursue the study of spectral properties of this operator.
Our aim is to study the asymptotics of the so called Integrated Density of States (IDS) of the operator \(H_\beta\) for energies near the bottom of the spectrum. The IDS is a function on the spectrum of the operator that computes the average number of eigenvalues per unit volume.
In disordered systems like the Anderson model with independent random variables, the IDS exhibits an exponential decay near the spectral edges at arbitrary dimension, known as Lifshitz tails. This is in stark contrast with the behavior of the IDS in periodic systems. The Lifshitz behavior of the IDS is a key ingredient to prove localization for random operators, although it is not a necessary condition (see e.g. Delone-Anderson models for which the IDS might not even exist but localization still holds \cite{rojas2013anderson,elgart2014ground} ). The connection between the IDS behavior at the bottom of the spectrum and localization explains the important role played by the IDS in the
spectral and dynamical study of random Schr\"odinger operators.

In \cite{Sabot2019}, the authors conjecture that  the asymptotic behavior of the IDS of the random Schr\"odinger
operator \(H_\beta\) appearing in connection to reinforced random processes does not exhibit Lifshitz tails.
This is due to dependencies in the random variables, that
imply that the bottom of the spectrum is not attained by extreme values
of the random variables, but can be attained by several configurations of
the potential.

In this article we show that the IDS \(N (E,H_{\beta })\) of the operator \(H_{\beta }\) does not exhibit Lifshitz tails, and undergoes a phase transition  \ in its behavior as a function of \(E\), depending on the dimension and the strength of the disorder. This follows from a phase transition in the associated reinforced random process and supersymmetric
sigma-model. Namely, we prove that in dimension one, for any value of the disorder strength, the IDS behaves roughly as  \(\sqrt{E}\) as \(E\downarrow 0\), while in dimension two and above, this behavior holds for large disorder.
On the contrary, in dimension three and above the decay rate is bounded above by \(E\) at weak disorder.

To the best of our knowledge, the operator \(H_\beta\) is the first Anderson-type model for which the IDS is known to undergo a phase transition,
whose dependence on the disorder strength and dimension is similar to the one in the metal-insulator transition  conjectured for the Anderson model.
Note that the transitions appearing in the literature for the IDS of Anderson-type models (the so-called  classical-quantum transitions) are transitions in the exponents of the Lifshitz tails depending on the decay of the single site potential \cite{Leschke-Warzel-2004}\cite{Fukushima-Ueki2010}.
A phase transition which does not
involve Lifshitz tails has been observed in the IDS for certain random spin models \cite{ErdoesSchoeder2014}.
As far as we know, the operator \(H_\beta\) provides a first physically motivated example where Lifshitz tails break down, even in presence of pure point spectrum. The latter contributes to the family of very specific models for which the violation of Lifshitz tails is known \cite{kirschnitzscher,najar09,kloppnakamura09,bakerlossstolz09,mulsto}.

\bigskip

We proceed to define the random Schr\"odinger operator \(H_\beta\) on \(\mathbb{Z}^d\) appearing in connection with the hyperbolic \(H^{2|2}\) sigma-model and reinforced random processes. Let \(\mathbb{Z}^d\) be
the undirected square lattice, with vertex set \(V(\mathbb{Z}^d)\) and edge set \(E(\mathbb{Z}^d)\).
By abuse of notation we will often identify the set in \(\mathbb{Z}^d\) with its vertex set and, in particular, write \(\mathbb{Z}^d\) instead of \(V(\mathbb{Z}^d).\)
The operator \(H_\beta\) is defined as follows: let \(W_e = W_{i,j} > 0\) be the edge weight of \(e=\{i,j\}\) on \(\mathbb{Z}^d\), and \(P^W\) be the associated adjacency operator of \(\mathbb{Z}^d\), or equivalently, \(P^W\) is the operator on \(\ell^2(\mathbb{Z}^d)\) defined by
\begin{equation}
P^W f(i)=\sum_{j:j\sim i}W_{i,j}f(j),\ \forall f\in \ell^2(\mathbb{Z}^d),
\end{equation}
where \(j\sim i\) means that \(\{i,j\}\) is an edge of the lattice \(\mathbb{Z}^d\).
We consider \(\mathcal{H}_{\beta }\in \mathbb{R}^{\mathbb{Z}^d\times \mathbb{Z}^d}\), the
infinite symmetric matrix defined by 
\begin{equation}\label{eq:defHbetazd}
    \mathcal{H}_{\beta }:= 2\beta -P^{W},
\end{equation}
where \(\beta\) is a
diagonal matrix whose diagonal entries \((\beta_i)_{i\in \mathbb{Z}^d}\) form a family of positive random variables defined as follows (c.f. Theorem 1 of \cite{Sabot2017} and Proposition 1 of \cite{Sabot2019}):  \(\forall i\in \mathbb{Z}^d\), \(\beta_i >0\) a.s. and for all sub lattice \(\Lambda \subset  \mathbb{Z}^d\) finite, the Laplace transform of \((\beta_i)_{i\in \Lambda}\) equals
\begin{align}\label{eq-laplace-beta-Lambda}
&\mathbb{E}^W(e^{-\left< \lambda,\beta \right>_{\Lambda}})\\
&=
e^{-\sum_{i,j\in \Lambda, i\sim j}W_{i,j}( \sqrt{(1+\lambda_i)(1+\lambda_j)}-1)-\sum_{i\in \Lambda,j\notin \Lambda, i\sim j}W_{i,j} (\sqrt{1+\lambda_i}-1)}
\prod_{i\in \Lambda} \frac{1}{\sqrt{1+\lambda_i}}.\nonumber
\end{align}
The law of this random field \(\beta\) is characterized by the above Laplace transform. This Laplace transform is a particular case of the general version, given in equation \eqref{eq-Laplace-nuGWthetaeta}, with \(\theta\equiv 1\). By means of the Laplace transform \eqref{eq-laplace-beta-Lambda}, one can see that if \(i\) and \(j\) are not related by an edge in \(\mathbb{Z}^d\), then \(\beta_i\) and \(\beta_j\) are independent. We say that the field \(\beta\) is 1-dependent. The infinite-volume distribution of \((\beta_i)_{i\in\mathbb{Z}^d}\) will be denoted by
\(\nu^W\) and the associated expectation is denoted by \(\mathbb{E}^W\). The operator \(\mathcal{H}_\beta\) defined by
\begin{equation}
\label{eq-def-H-wholeZd}
\begin{aligned}
\mathcal{H}_\beta f(i) = 2 \beta_i f(i)-\sum_{j:j\sim i}W_{i,j} f(j),\ \forall i\in \mathbb{Z}^d,
\end{aligned}
\end{equation}
maps \(\mathcal{D}\to \ell^2(\mathbb{Z}^d)\) almost surely, where \(\mathcal{D}\subset \ell^2(\mathbb{Z}^d)\) is the set of sequences with finite support, which is dense.

In this paper, we will set all \(W_{i,j}\) equal, and in an abuse of notation denote this common value \(W\) too (as it will not cause any ambiguity in the sequel).  This condition ensures the operator \(\mathcal{H}_{\beta}\) is ergodic w.r.t. the translations in \(\mathbb{Z}^d\). 
Ergodicity is a key ingredient to prove the existence of the IDS and the deterministic nature of the spectrum of
\(\mathcal{H}_\beta\)   using standard arguments.
In this case the spectrum is given by \(\sigma(\mathcal{H}_\beta) =  [0,+\infty)\)
(see \cite[Theorem 2.(i)]{Sabot2019} and  \cite{theserapenne}).
At times, we still write \( W_{i,j} \) to specify the vertex $i$ that we are considering or to emphasize
the generality of the probability measures.

By Proposition 1 of \cite{Sabot2017} or Lemma 4 of \cite{Sabot2019}, any finite marginal \((\beta_i)_{i\in \Lambda}\) (i.e. \(\Lambda  \subset\mathbb{Z}^d\) is a finite subset) has the following explicit probability density w.r.t. the product Lebesgue measure \(d \beta = \prod_{i\in \Lambda}d \beta_i\):
\begin{equation}
\begin{aligned}
\label{eq-density-beta-Lambda}
\nu^{W,\eta^w}_{\Lambda} (d\beta )=\mathbbm{1}_{\mathcal{H}_{\beta,\Lambda} > 0}
e^{-\frac{1}{2}\Big (\left< 1,\mathcal{H}_{\beta,\Lambda}1 \right>+\left< \eta^w,\mathcal{H}_{\beta,\Lambda}^{-1} \eta^w \right>-2\left< 1,\eta^w \right>\Big )}\frac{1}{\sqrt{\det \mathcal{H}_{\beta,\Lambda}}} \left( \frac{2}{\pi} \right)^{|\Lambda|/2} d\beta.
\end{aligned}
\end{equation}
The corresponding average will be denoted by \(\mathbb{E}_{\Lambda}^{W,\eta^w}\).
Here \(\eta^w\) is a vector denoting a wired boundary condition on \(\Lambda\), defined by
\begin{equation}
\begin{aligned}
\label{eq-eta-boundary-Lambda}
\eta^w(i):=\eta^w_{\Lambda}(i)=\sum_{j\notin \Lambda, j\sim i}W_{i,j}, \quad\mbox{for}\, i\in \Lambda,
\end{aligned}
\end{equation}
and
\begin{equation}\label{def:H-finite-volume}
\mathcal{H}_{\beta,\Lambda}:=(2 \beta - P^W)_{|\Lambda}={1}_{\Lambda}\mathcal{H}_{\beta } {1}_{\Lambda}  
\end{equation}
is the operator \(\mathcal{H}_\beta\) restricted on the set \(\Lambda\) with simple boundary condition, i.e. a finite matrix defined by
\begin{equation}
\mathcal{H}_{\beta,\Lambda} f(i):=2 \beta_i f(i)-\sum_{j\in \Lambda:j\sim i}W_{i,j} f(j),\qquad \forall f\in \mathbb{R}^{\Lambda }.
\end{equation}
Here \({1}_{\Lambda}\) is the projection operator on \(\Lambda\). 
Note that even if we replace \(\eta^w\) by an arbitrary \(\eta \in \mathbb{R}_{\ge 0}^\Lambda\) , \eqref{eq-density-beta-Lambda} is still a probability density. In this case the probability will be denoted by \(\nu^{W,\eta}_\Lambda\) and the expectation by \(\mathbb{E}^{W,\eta}_\Lambda\), to stress the \(\eta\) dependence. A more general finite volume density is given in Theorem  \ref{thm-the-multivariate-inverse-gaussian-distribution} below.

 Sometimes we will write \(\mathcal{H}_{\beta,\Lambda}^S:=\mathcal{H}_{\beta,\Lambda}\) to insist on the type of boundary conditions considered, that we call simple boundary condition.  We will also consider the operator with Dirichlet boundary condition, which will be denoted by \(\mathcal{H}_{\beta,\Lambda}^D\) and is defined  by
\begin{equation}
\begin{aligned}
\mathcal{H}^D_{\beta, \Lambda}:= (2\beta -P^{W} )_{\Lambda}+WM_{2d-n}=
\mathcal{H}_{\beta, \Lambda}+WM_{2d-n},
\end{aligned}
\end{equation}
where \(M_{2d-n}\) is the multiplicative operator by \(2d-n\) acting on \(\ell^2\left(\mathbb{Z}^d\right)\), where for every \(i\in \Lambda\), \(n_i:=\operatorname{deg}(i)\) in \(\Lambda\), i.e. \(n_i=\sum_{j\in \Lambda, \ j\sim i}1\).

In the usual Anderson model, the random Schr\"odinger operator \(H=-\Delta + \lambda V\) with a bounded potential, the edge weight
equals 1 (in the discrete Laplacian \(\Delta\), entries are 0 or 1), and the disorder parameter \(\lambda>0 \) modulating the intensity of the random potential allows for two well-defined regimes, that of strong disorder (\(\lambda\gg 1 \)) and that of weak disorder (small  \(\lambda  \)). In \(\mathcal{H}_\beta\), however, the edge
weight equals \(W\), and the law of the random potential depends also on \(W\), hence the disorder parameter does not appear as a coupling constant but is encoded in the law of \(\beta\).
To have an expression that resembles the Anderson model we consider the rescaled operator \({H}_{\beta }\) defined by 
\begin{equation}
{H}_{\beta }:= \frac{1}{W}\mathcal{H}_{\beta}= \frac{2\beta }{W}-P=(-\Delta )+\left(\frac{2\beta }{W}-2d\right) ,\qquad P_{ij}:=\mathbbm{1}_{i\sim j}.
\end{equation}
The corresponding finite volume operator with Dirichlet boundary condition is then 
\begin{equation}\label{def:HD}
H^D_{\beta, \Lambda}:= \frac{1}{W}\mathcal{H}^{D}_{\beta,\Lambda }=  \Big (\frac{2\beta }{W}-P \Big)_{\Lambda}+M_{2d-n}=H_{\beta, \Lambda}+M_{2d-n},
\end{equation}
where \(H_{\beta, \Lambda}=H_{\beta, \Lambda}^{S}\) is the operator with simple boundary condition.

Note that, by the explicit Laplace transform
\eqref{eq-laplace-beta-Lambda} (c.f. \cite[Theorem C]{Collevecchio2018}), we have for every \(j\in \mathbb{Z}^d\) and for every \(\lambda>0\),
\[
\mathbb{E}^{W}\left[ e^{-\lambda \beta_{j}} \right]= \frac{ e^{-2dW (\sqrt{1+\lambda }-1)}}{\sqrt{1+\lambda }}.
\]
Therefore, the one point marginal of the random
potential is known to be a reciprocal inverse Gaussian distribution.

It follows that  the mean of \(2 \beta_i\) is \(2d W +1\), and its
variance is \(2dW+2\). The corresponding  rescaled potential \(V_{i}:=\frac{2 \beta_i}{W}-2d\)
has mean \(\mathbb{E}^{W}[V_{i}]=\frac{1}{W}\) and variance \(Var[V_{i}]=\frac{2d}{W}+\frac{2}{W^{2}}.\)
Therefore, analogously to the case of the Anderson model, for \(H_\beta\) we can identify two regimes: \(W\) small corresponds to a strong disorder regime and \(W\) large, to a weak disorder regime.
Indeed, for large \(W\) we have \(\mathbb{E}^{W}[V_{i}]=1/W\simeq 0\), 
and   \(Var[V_{i}]=O(1/W)\simeq 0\), hence \(H_{\beta }\) is a small perturbation of 
\(2d-P=-\Delta \).  On the contrary, for small \(W\)
both mean and variance are large, 
 \(\mathbb{E}^{W}[V_{i}]=1/W,\)  and \(Var[V_{i}]\simeq 1/W^{2},\) hence \(H_{\beta }\)
 is dominated by the diagonal disorder.
\bigskip

Our main object of study is the Integrated Density of States (IDS) \(N(E)=N (E,H_{\beta })\)  for \(H_\beta\) at an energy \(E\in \mathbb R\), defined by 
\begin{equation}\label{eq:Nviafinitevolume}  
N (E,H_{\beta})=\lim_{L\to \infty} \mathbb{E}^{W}_{\Lambda_{L} }\left[N_{\Lambda_L}(E,{H}^{\#}_{\beta ,\Lambda_{L}}) \right],
\end{equation}
where \(\Lambda_L\) is a box in \(\mathbb Z^d\) of side \(2L+1\) centered at zero
and \(N_{\Lambda_L}(E,{H}^{\#}_{\beta ,\Lambda_{L}})\) is the finite volume IDS on \(\Lambda_{L}\) 
defined by
\begin{equation}\label{def:finite-vol-IDOS}
N_{\Lambda_L}(E,{H}^{\#}_{\beta ,\Lambda_{L}}):= \frac{1}{|\Lambda_L|}
\sum_{\lambda \in \sigma \left ({H}^{\#}_{\beta,\Lambda_{L}}\right )\cap (-\infty,E]}\hspace{-0,7cm}1\ =\  \frac{1}{|\Lambda_L|}
\operatorname{tr} \Big (\mathbbm{1}_{(-\infty,E]}  ({H}^{\#}_{\beta ,\Lambda_{L}} ) \Big).
\end{equation}
Here  \(\# \in \{D,S\}\) indicates if we have Dirichlet (see \eqref{def:HD}) or  simple boundary conditions.
Note that the usual definition of the IDS (see e.g. \cite[Corollary 3.16]{AW})
does not contain the expectation in the right-hand side of the equation \eqref{eq:Nviafinitevolume}. However the equivalence between these two definitions follows from, e.g., \cite[Lemma 4.12]{AW}. Also, the limiting function \(N\) does not depend on the boundary conditions in the finite-volume restriction of \(H_\beta\) to the box, so we can replace the simple boundary conditions with Dirichlet boundary conditions and the result still holds \cite[Lemma 4.12]{AW}).

\bigskip

We are interested in the asymptotics of \(N(E)\) for \(E \searrow 0\), that is, at the bottom of the spectrum of \(H_\beta\). For a random Schr\"odinger operator with i.i.d. random
potential, the Lifshitz tails estimate (e.g. \cite{Klopp2010}) claims that, near the bottom
of the spectrum (assuming it is 0), the integrated density of states behaves like
\begin{equation}\label{eq:expdecayN}
N(E)= c e^{-E^{-d/2+o(1)}}
\end{equation}
in \(d\) dimensions. This is in stark contrast with the case of the free laplacian which exhibits Van Hove asymptotics, that is, \(N(E)\asymp E^{d/2}\) (see \cite[Theorem 3]{Anves2009}).
The exponential decay in \eqref{eq:expdecayN} appears since, by the i.i.d. nature of the potential,
configurations near the bottom of the spectrum are highly unlikely.
Lifshitz tails also appear in models exhibiting correlations in the potential, for example
in potentials given by a linear combination of i.i.d. random variables \cite{KM2007}.
This behavior may be violated for example when the operator is not monotonous in the random variables
\cite{kirschnitzscher,najar09,kloppnakamura09,bakerlossstolz09}, or the lattice  is replaced by a random graph.
In particular the  random Laplacian of the percolation subgraph of bond-percolation with parameter \(p\) on
\(\mathbb{Z}^d\), exhibits a transition between Van Hove and Lifshitz behavior depending of the percolation parameter \(p\)
 \cite{mulsto,Muller2011}.

For our 
 \(H_\beta\) operator, the 1-dependence of the \(\beta\) variables entails that
 the number of realizations of the potential favoring low energy states is large and hence
Lifshitz tails do not occur.
Precisely,  we will prove the following three results.

\begin{theorem}[lower bound on the IDS]
\label{thm-main-theorem}
We define
\begin{equation}\label{def:Wcrit}
W_{cr}=W_{cr} (d):= \max\{W_{c},W_{c}' \},
\end{equation}
where \(W_{c}>0\)  (resp. \(W_{c}'>0\)) is the (dimensional dependent) parameter introduced in Theorem \ref{thm-localization-of-the-ground-state-green-function} (resp. Theorem \ref{thm-localization-fractional-moment}).  In particular \(W_{cr}=\infty \) for \(d=1.\)
Then, for each \( 0<W<W_{cr}\) there exist  constants \(c=c(W,d)>0\)  \(E_{0} =E_{0} (W,d,c)>0\)
such that
\begin{equation}
\begin{aligned}
\label{eq-NE-nolifhistz}
N(E,H_{\beta }) \ge c(-\log E)^{-d} \sqrt{E}, \qquad \forall 0<E<E_{0}.
\end{aligned}
\end{equation}

\end{theorem}
\vspace{0,2cm}
Actually, we will show  in Lemma \ref{lem:comparaison} below that  \(W_{c}' (d) >W_c (d)\)
for all  \(d\geq 2\). Equation \eqref{eq:Wc'} yields   $W_{c}' (d)\leq 0.1$.  

The next result concerns the regularity of the finite volume IDS with simple/Dirichlet boundary condition
defined in \eqref{def:finite-vol-IDOS} above.

\begin{theorem}[Wegner type estimate]
\label{wegner-type-theorem} \hspace{2cm}

For all \(W>0\)  the 
finite volume IDS \(N_{\Lambda_L}(E,H_{\beta,\Lambda_L}^\#)\), with   \(\# \in \{D,S\}\) satisfies the bound
\begin{equation}\label{eq:wegner-17}
\begin{aligned}
\mathbb{E}^{W,\eta^{w}}_{\Lambda_{L}}\left[N_{\Lambda_L}(E+\varepsilon,H_{\beta,\Lambda_L}^{\#}) -N_{\Lambda_L}(E-\varepsilon,H_{\beta,\Lambda_L}^{\#})\right]\le    4\sqrt{\frac{W}{2\pi}} \sqrt{\varepsilon }
\end{aligned}
\end{equation}
uniformly in \(\Lambda_L\),  \(E\in \mathbb{R}\) and  \(\varepsilon > 0\).\vspace{0,2cm}

Moreover, for \(d\geq 3\) we define $W_{0}=W_{0} (d):= \max\{W_{0}',4^{8} \},$ where $W_{0}'=W_{0}' (d)\geq 1$ is the
parameter introduced in Theorem \ref{th:bound-u-fluctDSZ}.  Then for all \(W\geq W_{0},\) the following improved estimate holds
\begin{equation}\label{eq:wegner-18}
\begin{aligned}
\mathbb{E}^{W,\eta^{w}}_{\Lambda_{L}}\left[N_{\Lambda_L}(E+\varepsilon,H_{\beta,\Lambda_L}^{\#}) -N_{\Lambda_L}(E-\varepsilon,H_{\beta,\Lambda_L}^{\#})\right]\le    C\sqrt{W}\,  \varepsilon 
\end{aligned}
\end{equation}
uniformly in \(\Lambda_L\),  \(E\in \mathbb{R}\) and  \(\varepsilon > 0\), where \(C>0\) is some constant.
\end{theorem}


\begin{theorem}[upper bound and regularity for the IDS]\label{corol-upper-bound-on-IDS}
\hspace{2cm}

For all \(W>0\) the function \(E\mapsto N (E,H_{\beta})\) is H\"older continuous with exponent \(\frac{1}{2}\) and H\"older seminorm \([N]_{C^{0,\frac{1}{2}}}\leq 2\sqrt{W/\pi}.\) In particular, it satisfies the bound 
\begin{equation}
\begin{aligned}
\label{eq-NE-nolifhistz-upper-bound}
 N(E,H_{\beta}) \le  2\sqrt{\frac{W}{\pi}}\sqrt{E}\qquad \forall E>0.
\end{aligned}
\end{equation}
Moreover,  for \(d\geq 3\) and  \(W\geq W_{0},\)
the function \(E\mapsto N (E)\) is Lipschitz continuous with  Lipschitz constant
\(Lip (N)\leq   C\sqrt{W}/2, \) where \(C\) is the constant introduced in Theorem
\ref{wegner-type-theorem}.
In addition, it satisfies the bound
\begin{equation}
\begin{aligned}
\label{eq-NE-nolifhistz-upper-bound-d=3}
 N(E,H_{\beta}) \le C' E\qquad \forall E>0,
\end{aligned}
\end{equation}
for some constant \(C'>0\) independent of \(W.\)
\end{theorem}

\paragraph{Discussion on the results and open questions.}\hspace{2cm}

\paragraph{Behavior of the IDS near $\mathbf{E=0}$.}
Theorems \ref{thm-main-theorem} and  \ref{corol-upper-bound-on-IDS}
imply, in the strong disorder regime \(W<W_{cr}\) and for all dimension $d$,
\begin{equation}
 c \frac{1}{|\ln E|^{d}}\sqrt{E}\leq N(E) \le  2\sqrt{\frac{W}{\pi}} \sqrt{E}
\end{equation}
as \(E\searrow 0.\)
In particular, our results indicate that for such a model, that is ergodic and features a 1-dependent
random potential, Lifshitz tails do not emerge.
Note that it is conjectured in \cite{Sabot2019} that, in the case of constant weights,
the asymptotic behavior of the IDS of the operator \(H_\beta\)  is \(\sqrt{E}\).
Therefore we prove this conjecture in the strong disorder regime up to a logarithmic correction.
We anticipate that this result will also apply to non-constant weights, provided they do not vary
excessively and the relevant quantities can be defined.

Note that, for strong disorder \(W\ll 1\), the random variables \(\beta\)
 are approximately iid with Gamma distribution (cf. the explicit Laplace transform in \eqref{eq-laplace-beta-Lambda}).
Therefore one expects that
\[
N (E,H_{\beta })=N(EW,\mathcal{H}_{\beta})\simeq \mathbb{P} (2\beta_{0}<EW)\propto \sqrt{EW}
\]
for \(EW<1,\) which is indeed what we obtained.

For weak disorder  \(W\gg 1\) the random variables approach the constant value \(2d,\) hence one expects convergence to the IDS of \(2d-P=-\Delta \) as \(W\to \infty\) which is proportional to  $E^{\frac{d}{2}}.$
Note that the improved bound \eqref{eq-NE-nolifhistz-upper-bound-d=3} is compatible with this expectation
and moreover  shows the IDS undergoes a phase transition at \(d\geq 3.\)
A more precise comparison with \(-\Delta\) would require also a lower bound for \(N (E,H_{\beta })\)
at weak disorder, which  is still missing. This leads to the following open question.
\vspace{0,2cm}

\textbf{Open question 1:} what is the exact asymptotics of the IDS at the bottom of the spectrum
at weak disorder  \(W\gg 1\), depending on the  dimension?
Does it approach the IDS of the Laplacian?

\paragraph{Critical value of the disorder strength $W$.}
For dimension $d\geq 3,$ we have proved a phase transition for the IDS in the following sense: for $E$ near zero 
 \(N(E)\approx\sqrt{E}\) when \(W\leq W_{c}' (d)\leq 0.1\) and \(N(E)\leq  C'E\ll\sqrt{E} \) when \(W\geq W_{0}> 1\).
 A natural question is whether this phase transition  is unique. This would require to know 
the behaviour of \(N(E)\) for intermediate values of \(W\), which is at the moment out of reach for our techniques.
This unicity is known in the case of the vertex reinforced jump process: in  \cite{Poudevigne2019} Poudevigne proved 
 that there is a unique transition point \(W^*(d)\) between recurrence and transience for the vertex reinforced jump process
 with constant weights \(W\) on \(\mathbb{Z}^d\) when \(d\geq 3\). The proof uses  a  monotonicity property concerning
 \(\mathcal{H}_{\beta}^{-1}\), that follows from a clever coupling  (see also Theorem \ref{th:monotonicity}  in the appendix), and 
  the 0-1 law in \cite[Proposition 3]{Sabot2019}.
As the operator \(\mathcal{H}_{\beta}\) is crucial to define the random environment of the vertex reinforced jump process, the following question arises. \vspace{0,2cm}

\textbf{Open question 2:}  is the phase transition for the density of states of \(H_{\beta}=\mathcal{H}_{\beta}/W\) unique and  does it occur at the same value \(W^*(d)\)? \vspace{0,2cm}

The strong disorder behavior  \(N(E)\approx \sqrt{E}\) corresponds to the recurrence region in the vertex reinforced jump process. A sufficient condition to obtain this behavior  is
  \[\mathbb{E}_{\Lambda}^{W,\eta^w}\left[\left|\mathcal{H}_{\beta,\Lambda}^{-1}(0,i)\right|^{s}\right]\le e^{-c |i|}\]
  for some \(s\in (0,1)\) uniformly in \(\Lambda\). At the moment we only have this bound for \(W<W'_{c}(d)\leq 0.1\).

On the other hand,  a sufficient condition to obtain the weak disorder behavior \(N(E)\leq C'E\) , which
 corresponds to the transience region in the vertex reinforced jump process,
is to prove the \(L^2\)-integrability of the martingale \((\psi_L)_{L\in \mathbb{N}}\), defined in Section \ref{sec:alterapproach}. For \(W > W_0(d)\) and \(d\ge 3\), this bound follows from Lemma~\ref{le-ward}.
Note that, if \(L^2\)-integrability was equivalent to transience  on $\mathbb{Z}^{d},$
then the bound  \(N(E)\leq C_{W}E\), for some constant $C_{W},$  would hold for all $W>W^*(d).$
In the case of trees, this equivalence was proved in  \cite{theserapenne}.
However on \(\mathbb{Z}^d\)  only uniform integrability of \((\psi_L)_{L\in\mathbb{N}}\)  for  $W>W^*(d)$ has been proven  so far (see \cite{theserapenne}). This result suggests another possible scenario, with 
an intermediate phase where the vertex reinforced jump process is transient but \((\psi_L)_{L\in\mathbb{N}}\)
is not bounded in \(L^2\) and may correspond to an intermediate phase for the behaviour of the density
of states of \(H_{\beta}\) too. This could  also be consistent with the presence of a phase where  \(H_{\beta}\)
exhibits singular continuous spectrum. The existence of singular  continuous spectrum for certain
random Schr\"odinger operators is conjectured to be true in high dimension by a growing 
community of physicists (see for example \cite{AK23}).

\paragraph{The Anderson transition.}
Sabot and Zeng conjecture in \cite{Sabot2019} that the phase transition in linearly reinforced random processes
between recurrent and transient regimes is related to the dynamical localization and delocalization transition of
\(H_\beta\). This provides an additional motivation for the study of this random operator, since the
localization-delocalization transition for random Schr\"odinger operators is a long-standing open problem in
the theory of disordered systems, going back to the seminal work of P.W. Anderson \cite{Anderson58}.
As in the standard Anderson model, the operator \(H_{\beta}\) exhibits dynamical localization
for strong disorder  $W\leq W' (d)$ \cite{Collevecchio2018}. This motivates the following question.
\vspace{0,2cm}

\textbf{Open question 3:} does  dynamical localization  for  \(H_{\beta}\) hold 
at spectral band edges for any fixed disorder parameter $W>0,$ in particular at weak disorder?
This is so far out of reach, since the stardard proof relies on the Liftschitz tails of the IDS,
which are absent in our case.

\paragraph{Theorem  \ref{thm-main-theorem} : strategy of the proof.}
We argue in three steps.

\textit{Step 1.} By standard arguments (see Section 3.1) we have, for all \(L\geq 1,\) 
\[
N (E,H_{\beta })\geq  \mathbb{E}^{W,\eta^{w} }_{\Lambda_{L}}\left[N_{\Lambda_L}(E,{H}^{D}_{\beta ,\Lambda_{L}}) \right]\geq
\frac{1}{|\Lambda_{L}|} \nu^{W,\eta^{w} }_{\Lambda_{L}}\left (
({H}^{D}_{\beta ,\Lambda_{L}})^{-1} (0,0)\geq \frac{1}{E} \right ),
\]
where \(H^{D}_{\beta,\Lambda_{L}}\) was defined in \eqref{def:HD} and
\(N_{\Lambda_L}(E,{H}^{D}_{\beta ,\Lambda_{L}})\)  is the finite volume IDS with Dirichlet boundary condition
defined in \eqref{def:finite-vol-IDOS}.
\vspace{0,2cm}

\textit{Step 2.} We have no direct information on the probability density of

\(({H}^{D}_{\beta ,\Lambda_{L}})^{-1} (0,0),\) but we do have detailed information on the distribution of
\(({H}^{S}_{\beta ,\Lambda_{L}})^{-1} (0,0)=H_{\beta ,\Lambda_{L}}^{-1} (0,0).\) 
In Section 3.2 we show that, for \(E\leq \frac{1}{2},\) \(W<W_{cr},\) and \(L>1\)  large enough
we have
\begin{equation}
\begin{aligned}
\left\{(H_{\beta,\Lambda_L}^S)^{-1}(0,0) > \frac{1}{E}\right\} \Rightarrow \left\{( H_{\beta,\Lambda_L}^D)^{-1}(0,0) > \frac{1}{2E} \right\},
\end{aligned}
\end{equation}
on a configuration set \(\Omega_{loc}\) of probability close to one, precisely \(1-e^{-\kappa L},\)
for some positive constant \(\kappa >0\). Hence
\[
N (E,H_{\beta })\geq \frac{1}{|\Lambda_{L}|} \nu^{W,\eta^{w} }_{\Lambda_{L}}\left (
\Omega_{loc}\cap \left\{  H_{\beta ,\Lambda_{L}}^{-1} (0,0) \geq \frac{1}{2E}  \right\}
\right ).
\]
\vspace{0,2cm}

\textit{Step 3.}
The conditional density of \(y:=1/\mathcal{H}_{\beta ,\Lambda_{L}}^{-1} (0,0) =W/H_{\beta ,\Lambda_{L}}^{-1} (0,0),\) knowing
\(\beta_{0^{c}}=(\beta_{j})_{j\in\Lambda_{L}\setminus \{0 \}},\) is denoted by \(d\rho_{a_{0}}\) and explicitly given
in \eqref{conditional-density-y}.  All dependence on \(\beta_{0^{c}}\) is contained in
the parameter \(a_{0},\) defined in \eqref{a-def}, which contains the two-point Green's function of the ground state of \(H_{\beta,\Lambda\setminus \{0\}}\).
Therefore
\[
 \mathbb{E}^{W,\eta^{w} }_{\Lambda_{L}}\left [
\mathbf{1}_{\Omega_{loc}} \mathbf{1}_{\left\{  H_{\beta ,\Lambda_{L}}^{-1} (0,0) \geq \frac{1}{2E}
\right\}}
\right ]=\mathbb{E}^{W,\eta^{w} }_{\Lambda_{L}}\left [
\int \mathbf{1}_{\Omega_{loc}} \mathbf{1}_{\left\{  y \leq  2EW\right\}}  d\rho_{a_{0} (\beta_{0^{c}})} (y)
\right].
\]
A subtle point is to show that we can choose \(\Omega_{loc}\) as intersection of two events
\(\Omega_{loc} =\Omega_{loc,0}\cap \Omega_{loc,1}\) where \(\Omega_{loc,0}\) is measurable w.r.t.
\(\beta_{0^{c}}\) and  \(\Omega_{loc,1}=\{y\geq e^{-\kappa L}W\}\) is measurable w.r.t.
\(y.\) As a result we can write
\[
\mathbb{E}^{W,\eta^{w} }_{\Lambda_{L}}\left [
\int \mathbf{1}_{\Omega_{loc}} \mathbf{1}_{\left\{  y \leq  2WE\right\}}  d\rho_{a_{0} (\beta_{0^{c}})} (y)
\right]=
\mathbb{E}^{W,\eta^{w} }_{\Lambda_{L}}\left [ \mathbf{1}_{\Omega_{loc,0}}
\int_{e^{-\kappa L}W}^{2WE}  d\rho_{a_{0} (\beta_{0^{c}})} (y) \right].
\]
The set \(\Omega_{loc,0}\) ensures that \(a_{0}(\beta_{0^{c}})\leq e^{-\kappa L/2}\) for
all \(\beta_{0^{c}}\in\Omega_{loc,0}. \) Then, using the explicit form of \(\rho_{a}\)
and the fact that \(\Omega_{loc,0}\) has probability close to one we 
get
\[
\mathbb{E}^{W,\eta^{w} }_{\Lambda_{L}}\left [ \mathbf{1}_{\Omega_{loc,0}}
\int_{e^{-\kappa L}W}^{2WE}  d\rho_{a_{0} (\beta_{0^{c}})} (y) \right]\geq
(1-c e^{-\kappa L}) \rho_{e^{-\kappa L/2}} (e^{-\kappa L}W\leq y\leq 2WE).
\]
The result now follows from a direct analysis of the one-dimensional measure \(\rho_{a}.\)
The details are explained in Section 4.

\paragraph{Theorem \ref{wegner-type-theorem}: strategy of the proof.}
Note that \(H_{\beta,\Lambda_{L}}\pm 2\varepsilon =H_{\beta\pm \varepsilon ,\Lambda_{L}}.\) Following a standard argument we construct a sequence of potentials interpolating between
\(\beta+ \varepsilon\) and \(\beta-\varepsilon,\) by switching \(\varepsilon \) one site at the time.
As a result we obtain the following estimate (cf.  Lemma \ref{le:deltabound})
\begin{align*}
&\mathbb{E}^{W,\eta^{w}}_{\Lambda_{L}}\left[N_{\Lambda_L}(E+\varepsilon,H_{\beta,\Lambda_L}^{\#}) -N_{\Lambda_L}(E-\varepsilon,H_{\beta,\Lambda_L}^{\#})\right]\\
&\leq \frac{4}{|\Lambda_{L}|} \sum_{j\in \Lambda }
\mathbb{E}_{\Lambda_{L}}^{W,\eta^{w}}\left[\mathscr{L}_{\rho_{a_{j} (\beta_{j^{c}})}}(4W\varepsilon)  \right],\qquad \forall L>1,
\end{align*}
where \(\mathscr{L}_{\rho_{a_{j}} }\) denotes the L\'evy concentration (defined in \eqref{levy-concentrtion})
of the conditional measure \(\rho_{a_{j}}.\) Using the explicit formula for \(\rho_{a}\) we then show that \(\mathscr{L}_{\rho_{a} } (\varepsilon )\leq c \sqrt{\varepsilon }\) for some constant \(c>0\) independent of \(a\) (cf. \eqref{eq:levy-bound}), which gives the first result. For \(d\geq 3\) we bound the conditional density pointwise by 
\[
 \rho_{a} (y) \leq \frac{1}{\sqrt{2\pi} }
 \left(\frac{1}{a}+\frac{1}{\sqrt{a}} \right)\qquad \forall y>0.
\]
The result now follows from the following bound (cf.  Lemma \eqref{le:boundH0}).
\[
\mathbb{E}_{\Lambda_{L}}^{W,\eta^{w}}\left[1/a  \right]=
\mathbb{E}_{\Lambda_{L}}^{W,\eta^{w}}\left[e^{-u_{0}} \mathcal{H}_{\beta ,\Lambda_{L}}^{-1} (0,0) \right]\leq \frac{C_{d}}{W},
\]
where \(C_{d}>0\) is a constant depending only on the dimension.

\paragraph{Theorem \ref{corol-upper-bound-on-IDS}: strategy of the proof.}
The regularity bounds  follow directly from the Wegner estimate and \eqref{eq:Nviafinitevolume} by replacing \(E\) and \(\varepsilon \) by \(E/2.\)
On the contrary, the improved upper bound \eqref{eq-NE-nolifhistz-upper-bound-d=3} is proved in Proposition \ref{prop:upper-b-IDOSd=3}
using properties of the infinite volume distribution of \(\beta.\)

\paragraph{Organization of this paper.} In Section 2 we review some definitions and known results, and derive the modifications of these results that will be needed in the rest of the paper. A few additional technical results that we will also use in this section are summarized in  Appendix A. Section 3  covers the
first two steps in the proof of Theorem  \ref{thm-main-theorem}.
The final step is worked out in Section 4.
Section 5  contains the proof of Theorem   \ref{wegner-type-theorem}.
Some estimates on the $u$ field associated to the  \( H^{2|2} \)-model (see  \eqref{def:umeasure} below) which are necessary for the proof, but are also interesting in their own, are collected in Appendix B.
Note that all these proofs involve only properties of the \textit{finite volume} marginal distribution
but some of the above results can be recovered by exploiting properties
of the \textit{infinite volume} distribution.  The main ideas  are sketched in Section 6, while the detailed construction  can be found in  \cite{theserapenne}.
This alternative approach also provides the improved bound  \eqref{eq-NE-nolifhistz-upper-bound-d=3} in Theorem  \ref{corol-upper-bound-on-IDS}.
All corresponding details are given in Section 6.

\paragraph{Acknowledgements.} We would like to thank the anonymous referee for carefully
reading our manuscript and for suggestions that led to an improvement of the article.
V. Rapenne would like to thank Christophe Sabot for
suggesting working on this topic and for his valuable guidance. M. Disertori and
C. Rojas-Molina gratefully acknowledge the Deutsche Forschungsgemeinschaft
for support through the project RO 5965/1-1 AOBJ 654381.
X. Zeng and C. Rojas-Molina gratefully acknowledge the Agence Nationale de la Recherche for
financial support via the ANR grant RAW ANR-20-CE40-0012-01. 
M. Disertori also acknowledges support from the Hausdorff Center for Mathematics
(project ID 390685813) and the hospitality of the
Institute for Advanced Study in Princeton, where part of this work has been
carried out.

\section{Some previous results on the \(\mathcal{H}_\beta\) operator }
\label{sec:orgbf5a0a9}

As mentioned in the introduction, the \(\mathcal{H}_\beta\) operator has been studied in the literature in connection with linearly reinforced random processes and the \(H^{2|2}\) sigma-model. In this section we collect some tools and results that we will use in the next sections.
The following theorem can be found in \cite{Sabot2017,Sabot2019,Letac2017}.

\begin{theorem}[The multivariate inverse Gaussian distribution]
\label{thm-the-multivariate-inverse-gaussian-distribution}
Let \(\mathcal{G}=(V,E)\) be a finite graph.  For any \(W\in  \mathbb{R}_{> 0}^E\),
\(\theta \in \mathbb{R}_{ > 0}^V\) and \(\eta\in \mathbb{R}_{\ge 0}^V\), the following holds
\begin{equation}
\begin{aligned}
\label{eq-multi-IG-integral=1}
\int_{\mathcal{H}_{\beta,V} > 0} e^{-\frac{1}{2}\left( \left< \theta,\mathcal{H}_{\beta,V} \theta \right>+\left< \eta, \mathcal{H}_{\beta,V}^{-1}\eta \right>- 2\left< \theta,\eta \right> \right)} \frac{\prod_i \theta_i}{\sqrt{\det \mathcal{H}_{\beta,V}}}\left( \frac{2}{\pi} \right)^{|V|/2} d \beta =1,
\end{aligned}
\end{equation}
where  \(\mathcal{H}_{\beta,V}:= 2 \beta -P^W\in \mathbb{R}^{V\times V}\).
We denote by \(\nu^{W,\theta,\eta}_{\mathcal{G}}\) the probability defined by the above integral, in particular, \(\nu^{W,\eta}_{\mathcal{G}}=\nu^{W,\theta, \eta}_{\mathcal{G}}\) with \(\theta\equiv 1\). The associated expectations are denoted by \(\mathbb{E}^{W,\theta,\eta}_{\mathcal{G}}\) and the Laplace transform is given by 
\begin{equation}
\begin{aligned}
\label{eq-Laplace-nuGWthetaeta}
&\mathbb{E}_{\mathcal{G}}^{W,\theta,\eta}\left( e^{-\left< \lambda,\beta \right>} \right) =\\
& e^{-\sum_{i,j\in V,i\sim j}W_{i,j}
\big ( \sqrt{(\theta_{i}^2+\lambda_i)(\theta_{j}^2+\lambda_j)}-\theta_{i} \theta_{j}\big )-\sum_{i\in V}\eta_{i}\big  (\sqrt{\theta_{i}^2+\lambda_i}-\theta_{i}\big )}
\prod_{i\in V} \frac{{\theta_{i}}}{\sqrt{\theta_{i}^2+\lambda_i}}.
\end{aligned}
\end{equation}
Moreover, if \((\beta_i)_{i\in V}\) is distributed according to {\(\nu_{\mathcal{G}}^{W,\theta,0}\)}, and
\(\mathcal{G}'= (V',E')\) is the subgraph obtained by taking \(V'\subset V\) and \(E':= \{\{i,j \}\in E|\ i,j\in V' \},\)
then  the marginal law of \((\beta_i)_{i\in V'}\) is \(\nu_{\mathcal{G}'}^{W',\theta',\eta}\) where \(W',\theta'\) equal
\(W,\theta\) restricted on \(\Lambda'\), and \(\eta\) is defined by
\(\eta_i=\sum_{j\in V\backslash V'}W_{i,j}\theta_{j} \). Note that \(\eta\) is a generalization of \(\eta^w_{V'}\) defined in \eqref{eq-eta-boundary-Lambda}.
\end{theorem}

In the sequel of the paper, we will always assume that \(\theta\equiv 1\) and we will use the notation \(\nu_{\mathcal{G}}^{W,\eta}\) instead of \(\nu_{\mathcal{G}}^{W,1,\eta}\).

\paragraph{Remarks on various marginals.}
Note that, (see e.g. Remark 3.5 of \cite{Collevecchio2018}), in the case \(\eta \equiv 0\),  for any \(i\in V\),
\begin{equation}\label{eq-gamma-1over2G}
\begin{aligned}
\gamma:=\frac{1}{2 \mathcal{H}_{\beta,V}^{-1}(i,i)} \text{ has density }\mathbbm{1}_{\gamma > 0} \frac{1}{\sqrt{\pi \gamma}}e^{-\gamma},
\end{aligned}
\end{equation}
i.e. it's a Gamma random variable of parameter \(1/2\). This holds for any \(W\) and any finite {graph} \(\mathcal{G}\).

If we consider a box \(\Lambda_{L} \) of side \(2L+1\) in \(\mathbb{Z}^{d}\), a Borel function \(f:\mathbb{R}^{\Lambda_L}\to \mathbb{R}\), and \(\Lambda_L \subset \Lambda'  \subset \subset  \mathbb{Z}^d\) then
\begin{equation}\label{eq:marginal1}
\mathbb{E}^{W,\eta_{\Lambda_L}^{w}}_{\Lambda_{L} } \left[f\left(\beta_{\Lambda_L}\right)\right]=  \mathbb{E}^{W,0}_{\Lambda_{L+1} }\left[f\left(\beta_{\Lambda_L}\right)\right]=\mathbb{E}^{W,0}_{\Lambda_L\cup\{\delta\}} \left[f\left(\beta_{\Lambda_L}\right)\right]=\mathbb{E}_{\Lambda'}^{W,\eta_{\Lambda'}^{w}} \left[f\left(\beta_{\Lambda_L}\right)\right],
\end{equation}
where  \(\eta^{w}_{\Lambda_L}, \ \eta^{w}_{\Lambda'}\)
are given in  \eqref{eq-eta-boundary-Lambda}, the graph \(\Lambda_L\cup \delta \) has vertex set \(\Lambda_L \cup \{\delta  \}\)
and edge set \(E (\Lambda_L )\cup \{\{i,\delta  \}|\ i\in \Lambda_L  \},\) and we
defined \(W_{i\delta }=\eta_{\Lambda_L}^w(i)\), \(\forall i\in \Lambda_L\). Note that \(\eta\) can be seen as the boundary condition of the law of
the random potential.  The case \(\eta\equiv 0\) is called zero boundary
condition.

\paragraph{Connection with the \(H^{2|2}\) model.} 
Let \(\mathcal{G}\) be the graph associated to a box \(\Lambda \) of \(\mathbb{Z}^{d}\)
and  \(\eta \in [0,\infty)^{\Lambda}\) with at least one strictly positive
component.  The following expression defines a probability measure for \(u\in \mathbb{R}^{\Lambda}\) (cf \cite{Disertori2010})
\begin{align}\label{def:umeasure}
&\mu^{W, \eta}_{\Lambda} (u )=\\
&
e^{-\sum_{i\sim j, {i,j\in \Lambda}} W_{ij}  ( \cosh (u_{i}-u_{j})-1)}
e^{-\sum_{j\in \Lambda} \eta_{j} (\cosh u_{j}-1)}
\sqrt{\det \mathcal{H}_{\beta (u),\Lambda} }
\ \tfrac{1}{\sqrt{2\pi }^{|\Lambda|}}  du_{\Lambda},
 \nonumber
\end{align}
where we defined
\begin{equation}\label{eq:betau}
2\beta_{i} (u)= \sum_{j\in \Lambda}W_{ij} e^{u_{j}-u_{i}}+\eta_{i} e^{-u_{i}}\qquad \forall i\in \Lambda.
\end{equation}
The corresponding average is denoted by \(\mathbb{E}_{u,\Lambda }^{W,\eta }.\) Note that
the measure \(\mu_{\Lambda}^{W,\eta} (u)\) is also the  effective bosonic
field measure in Section 2.3 of
\cite{Disertori2010a}, which is obtained as a marginal of the $H^{2|2}$ measure after inserting horospherical coordinates.

The next lemma connects  \(\mu_{\Lambda}^{W,\eta} (u)\)  with \(\nu^{W, \eta}_{\Lambda} (\beta )\) and can be found in Proposition 2 and Theorem 3 in \cite{Sabot2017}.

\begin{lemma}[connection to \(H^{2|2}\)]\label{le2vrai}
Let \(\mathcal{G}\) be the graph associated to a box \(\Lambda \) of \(\mathbb{Z}^{d}\)
and  \(\eta \in [0,\infty)^{\Lambda}\) with at least one strictly positive
component. It holds
\begin{equation}\label{eq:connection-u-beta}
\mathbb{E}_{\Lambda}^{W,\eta } [f (\beta_{\Lambda})]=
\mathbb{E}_{u,\Lambda}^{W,\eta }[ f (\beta_{\Lambda} (u)]
\end{equation}
for any function \(f\) integrable with respect to the measure \(\nu^{W, \eta}_{\Lambda}.\)

Moreover, remembering that  \(\nu^{W,\eta}_{\Lambda }\) corresponds to the marginal of  \(\nu^{W,0}_{\Lambda\cup \delta  }\) with \(W_{j,\delta }=\eta_{j}\)
\(\forall j\in \Lambda ,\) (cf. eq. \eqref{eq:marginal1}) it holds
\begin{equation}
\label{exp-u-beta}
e^{u_{i}}= \frac{ \mathcal{H}_{\beta,\Lambda\cup \delta }^{-1} (i,\delta ) }{\mathcal{H}_{\beta,\Lambda\cup \delta }^{-1} (\delta ,\delta ) },\qquad \forall i\in \Lambda,
\end{equation}
where the above fraction is independent of \(\beta_{\delta }\). In particular, we have
\begin{align}
\mathbb{E}_{u,\Lambda}^{W,\eta } [ f (u)]&=
                                           \mathbb{E}_{\Lambda\cup \delta }^{W,0 }[ f ( u (\beta ))]
                                           =\mathbb{E}_{\Lambda}^{W,\eta }[ f ( u (\beta_{\Lambda} )) ].
\end{align}

\end{lemma}
\vspace{0,2cm}

Note that, by  the resolvent identity we have, for all \(j\in \Lambda ,\)
\begin{align}
\mathcal{H}_{\beta,\Lambda\cup \delta}^{-1} (j,\delta )&=
0 +\sum_{k\in \Lambda }\mathcal{H}_{\beta,\Lambda}^{-1} (j,k ) W_{k,\delta } \mathcal{H}_{\beta,\Lambda\cup \delta }^{-1} (\delta ,\delta ) \\
&=
\mathcal{H}_{\beta,\Lambda\cup \delta}^{-1} (\delta ,\delta )\sum_{k\in \Lambda }\mathcal{H}_{\beta,\Lambda}^{-1} (j,k )\eta_{k}=
\mathcal{H}_{\beta,\Lambda\cup \delta }^{-1} (\delta ,\delta )\ (\mathcal{H}_{\beta,\Lambda}^{-1}\eta) (j)  \nonumber
\end{align}
and hence
\begin{equation}\label{eq:ujasHb1}
e^{u_{j}}=\frac{ \mathcal{H}_{\beta,\Lambda\cup \delta }^{-1} (j,\delta ) }{\mathcal{H}_{\beta,\Lambda\cup \delta }^{-1} (\delta ,\delta ) }= (\mathcal{H}_{\beta,\Lambda}^{-1}\eta) (j) .
\end{equation}
It follows
\begin{equation}\label{eq:ujasHb2}
\mathbb{E}_{u,\Lambda}^{W,\eta }\left [f\left(e^{u_{j}-u_{j'}} \right) \right]=
\mathbb{E}_{\Lambda}^{W,\eta }\left [   f\left(\frac{(\mathcal{H}_{\beta,\Lambda}^{-1}\eta) (j)}{(\mathcal{H}_{\beta,\Lambda}^{-1}\eta) (j')} \right) \right].
\end{equation}
for any function \(f,\) as long as the left and right hand side are well defined.
In the special case \(\eta_{j}=\eta \delta_{jj_{0}}\) (pinning at one point) the formula simplifies to
\begin{equation}\label{eq:ujasHb3}
e^{u_{j}-u_{j_{0}}}= \frac{\mathcal{H}_{\beta,\Lambda}^{-1} (j,j_{0})}{\mathcal{H}_{\beta,\Lambda}^{-1} (j_{0},j_{0})}
=  \frac{\mathcal{H}_{\beta,\Lambda}^{-1} (j_{0},j)}{\mathcal{H}_{\beta,\Lambda}^{-1} (j_{0},j_{0})}. 
\end{equation}

With these notations, we can translate Theorem 1 and 2 of \cite{Disertori2010} into the following.

\begin{theorem}[decay of the ground state Green's function (1)]
\label{thm-localization-of-the-ground-state-green-function}

Let \(\Lambda_{L} \) be a finite box of side \(2L+1\) in \(\mathbb{Z}^{d},\)
and  \(\eta \in [0,\infty)^{\Lambda}\) with at least one strictly positive
component. We define
\begin{equation}\label{I_{W}}
I_W:=\sqrt{W} \int_{-\infty}^{\infty} \frac{dt}{\sqrt{2\pi}} e^{-W(\cosh t-1)}.
\end{equation}
We define \(W_{c}>0\) as the unique solution of \(I_{W_c}e^{W_c(2d-2)}(2d-1)=1\).
Then  for all \(0 < W < W_c\), \(I_{W}e^{W(2d-2)}(2d-1)<1 \) and 
for \(d=1\) we have \(W_c=+\infty\).
Finally set \(C_{0}=C_{0} (W) :=2e^{2W} \left(1-I_{W}e^{W(2d-2)}(2d-1) \right)^{-1}.\)

Then the following hold.\vspace{0,2cm}

\((i)\) For all \(i,j\in \Lambda\) such that \(\eta_i>0,\ \eta_j >0\), we have
\begin{equation}\label{eqth6GFdecay}
\mathbb{E}^{W,\eta}_{\Lambda_{L}}[\mathcal{H}_{\beta,\Lambda_{L}}^{-1}(i,j)]\
\le C_0e^{\sum_{k\in \Lambda_{L}}\eta_{k}}\left( \eta_i^{-1}+\eta_j^{-1} \right) \left[ I_{W} e^{W(2d-2)} (2d-1) \right]^{|i-j|},
\end{equation}
where \(|i-j|\) is the graph distance between \(i,j\) on \(\mathbb{Z}^d\).\vspace{0,2cm}

\((ii)\) Assume there is only one pinning at \(j_{0}\in \Lambda_{L},\) i.e.
\(\eta_{j}=\eta \delta_{jj_{0}}.\) Then  \(\forall j\in \Lambda_{L}\),
\begin{equation}
\label{eq-4.2}
\begin{aligned}
\mathbb{E}_{\Lambda_{L}}^{W,\eta}\left[
\sqrt{\tfrac{\mathcal{H}_{\beta,\Lambda_{L}}^{-1}(j_{0},j)}{
\mathcal{H}_{\beta,\Lambda_{L}}^{-1}(j_{0},j_{0})}}\ \right] \ \le \ C_0
\left( I_W e^{W(2d-2)}(2d-1) \right)^{|j-j_{0}|}.
\end{aligned}
\end{equation}

\end{theorem}

\paragraph{How the results above follow from \cite{Disertori2010}}
Note that in  \cite{Disertori2010}
$\eta$ is called $\varepsilon$ and it is assumed that  \(\sum_{k\in \Lambda}\varepsilon_{k}\leq 1.\) This implies that
the term
\(e^{\sum_{k\in \Lambda}\varepsilon_{k}}\) in \cite[eq.(2.24)]{Disertori2010} is bounded by \(e^{1}\),
which is absorbed by the global constant $C_{0}$ in   \cite[eq.(1.18)]{Disertori2010}.
Also, although in \cite{Disertori2010} the global constant $C_{0}$ is not given explicitely,
it follows directly from the second line in  eq. (2.20) therein, so we have given its precise value here. 

Note that the statement $(ii)$ above is slightly different from the one in \cite[Th. 2]{Disertori2010}. Namely, in    \cite{Disertori2010}
the pinning point is set to $j_{0}=0$ and the observable is $e^{\frac{1}{2}u_{j}},$ while,
 by \eqref{eq:ujasHb3}, the observable in \eqref{eq-4.2} is $e^{\frac{1}{2} (u_{j}-u_{j_{0})}}.$
 Since the two observables differ only by a function of $u_{j_{0}}$ (the variable at the pinning point) the proof for this
 modified observable is identical to the one in   \cite{Disertori2010} up to the final line in 
 \cite[eq.(3.5)]{Disertori2010}, where the integral
 $I_{\varepsilon_{0}}$ is replaced by $\sqrt{\frac{\varepsilon_{0}}{2\pi }} \int_{\mathbb{R}} dt e^{-\frac{1}{2}t} e^{-\varepsilon_{j_{0}} (\cosh t-1)}=1.$
As a consequence the  bound in \eqref{eq-4.2} is written in terms of  the same constant $C_{0}$ used in \eqref{eqth6GFdecay}
and is independent from the pinning stregth $\varepsilon_{j_{0}}$ while the bound in  \cite{Disertori2010} does depend on
$\varepsilon_{j_{0}}$ via the function $I_{\varepsilon_{j_{0}}}.$ \medskip

\textbf{Remark}. The function \(W\mapsto I_{W}\) is  monotone increasing (cf. Remark 1 after Theorem 1 in 
 \cite{Disertori2010}). Therefore the function \(W\mapsto F_{d} (W):=I_W e^{W(2d-2)}(2d-1)\) is also monotone increasing
  and \(W_{c}\) is well defined. \vspace{0,2cm}

In this article we will use the following extension of Theorem \ref{thm-localization-of-the-ground-state-green-function}.

\begin{theorem}[decay with wired bc (1)]\label{cor:decayd=1}
Let  \(\Lambda_{L}\) be a finite box of side \(2L+1\) in \(\mathbb{Z}^{d}\).
We consider  \((\beta_i)_{i\in \Lambda_{L}}\sim \nu_{\Lambda_{L}}^{W,\eta^{w}}\)
where \(\eta^w_{\Lambda_{L}}\) is the wired boundary condition introduced in
\eqref{eq-eta-boundary-Lambda}. Let \(W_{c}\) be as in
Theorem \ref{thm-localization-of-the-ground-state-green-function} above.

For all \(0<W<W_{c}\), \(j,j_{0}\in \Lambda\) we have
\begin{equation}
\label{eq-decayd=1}
\begin{aligned}
\mathbb{E}_{\Lambda_{L}}^{W,\eta^{w}}\left[ \sqrt{\tfrac{\mathcal{H}_{\beta,\Lambda_{L}}^{-1}(j_{0},j)}{\mathcal{H}_{\beta,\Lambda_{L}}^{-1}(j_{0},j_{0})}}\ \right] \le C_0 e^{-\kappa|j-j_{0}|}.
\end{aligned}
\end{equation}
where  \(\kappa=\kappa(W,d):=-\log (I_We^{W(2d-2)}(2d-1))>0\), and   \(C_0=C_{0} (W)\) is the constant introduced in Theorem \ref{thm-localization-of-the-ground-state-green-function}.
In particular, in \(d=1\) \(W_{c}=\infty,\) hence the bound holds \(\forall \, W>0.\)
\end{theorem}

\begin{proof}
By \eqref{eq:marginal1} we have 
\[
\mathbb{E}_{\Lambda_{L}}^{W,\eta^{w}}\left[
\sqrt{\tfrac{\mathcal{H}_{\beta,\Lambda_{L}}^{-1}(j_{0},j)}{
\mathcal{H}_{\beta,\Lambda_{L}}^{-1}(j_{0},j_{0})}} \right]=
\mathbb{E}_{\Lambda_{L+1}}^{W,0}\left[ \sqrt{\tfrac{\mathcal{H}_{\beta,\Lambda_{L}}^{-1}(j_{0},j)}{\mathcal{H}_{\beta,\Lambda_{L}}^{-1}(j_{0},j_{0})}} \right].
\]
By a random walk representation (cf. Proposition 6 of \cite{Sabot2019}  and notations therein)
\[
\frac{\mathcal{H}_{\beta,\Lambda_L}^{-1}(j_{0},j)}{\mathcal{H}_{\beta,\Lambda_L}^{-1}(j_{0},j_{0})}=
\sum_{\sigma \in \overline{\mathcal{P}}_{jj_0}^{\Lambda_{L} }}\frac{W_{\sigma }}{(2\beta )_{\sigma }^-}
\]
where \(\overline{\mathcal{P}}_{j_{0}j}^{\Lambda_{L} }\) is the set of nearest neighbor paths from \(j_{0}\) to \(j\)
in \(\Lambda_{L} \) that visit \(j_{0}\) only once.  Moreover, for every path \(\sigma\),
\[W_{\sigma}=\prod\limits_{k=0}^{|\sigma|-1}W_{\sigma_k,\sigma_{k+1}}\text{ and }(2\beta)_{\sigma}^-=\prod\limits_{k=0}^{|\sigma|-1}(2\beta_{\sigma_k}).\]
It follows, for all \(j,j_{0}\in \Lambda_{L} ,\)
\[
\frac{\mathcal{H}_{\beta,\Lambda_{L}}^{-1}(j_{0},j)}{\mathcal{H}_{\beta,\Lambda_{L}}^{-1}(j_{0},j_{0})}
\leq \frac{\mathcal{H}_{\beta,\Lambda_{L+1}}^{-1}(j_{0},j)}{\mathcal{H}_{\beta,\Lambda_{L+1}}^{-1}(j_{0},j_{0})},
\]
since in the term in the right-hand side contains more paths. Hence
\[
\mathbb{E}_{\Lambda_{L}}^{W,\eta^{w}}\left[ \sqrt{\tfrac{\mathcal{H}_{\beta,\Lambda_{L}}^{-1}(j_{0},j)}{\mathcal{H}_{\beta,\Lambda_{L}}^{-1}(j_{0},j_{0})}} \right]=
\mathbb{E}_{\Lambda_{L+1}}^{W,0}\left[ \sqrt{\tfrac{\mathcal{H}_{\beta,\Lambda_{L}}^{-1}(j_{0},j)}{\mathcal{H}_{\beta,\Lambda_{L}}^{-1}(j_{0},j_{0})}} \right]
\leq \mathbb{E}_{\Lambda_{L+1}}^{W,0}\left[
\sqrt{\tfrac{\mathcal{H}_{\beta,\Lambda_{L+1}}^{-1}(j_{0},j)}{\mathcal{H}_{\beta,\Lambda_{L+1}}^{-1}(j_{0},j_{0})}} \right].
\]
By the monotonicity result \cite[Theorem 6]{Poudevigne2019} (cf.  Corollary \ref{corollary-monotonicity} in the appendix) we have, 
setting \(\eta_{j}=W\delta_{jj_{0}}\) \(\forall j\in \Lambda_{L+1},\) 
\begin{align}
& \mathbb{E}_{\Lambda_{L+1}}^{W,0}\left[
\sqrt{\tfrac{\mathcal{H}_{\beta,\Lambda_{L+1}}^{-1}(j_{0},j)}{\mathcal{H}_{\beta,\Lambda_{L+1}}^{-1}(j_{0},j_{0})}} \right]
\ \leq\  \mathbb{E}_{\Lambda_{L+1}}^{W,\eta }\left[ \sqrt{\tfrac{\mathcal{H}_{\beta,\Lambda_{L+1}}^{-1}(j_{0},j)}{\mathcal{H}_{\beta,\Lambda_{L+1}}^{-1}(j_{0},j_{0})}} \right]\\
&\qquad  =
 \mathbb{E}_{u,\Lambda_{L+1}}^{W,\eta }\left[ \sqrt{ e^{u_{j}-u_{j_{0}}}} \right]
 \leq  C_0 e^{-\kappa|j-j_{0}|}, \nonumber
\end{align}
with  \(\kappa=-\log (I_We^{W(2d-2)}(2d-1))>0\), and \(C_0=2e^{2W}\). In the last step we  used Lemma \ref{le2vrai} and 
 Theorem \ref{thm-localization-of-the-ground-state-green-function}(ii). 
\end{proof}\vspace{0,2cm}
\vspace{0,2cm}

Another useful result on the decay of the ground state Green's function is Theorem 2.1 of \cite{Collevecchio2018},
more precisely Equation (5.4) therein. We state this result for our applications

\begin{theorem}[decay of the ground state Green's function (2)]
\label{thm-localization-fractional-moment}
Let  \(\Lambda_L\) be a finite box of side \(2L+1\) in  \(\mathbb{Z}^d\),
and define
\begin{equation}\label{eq:Wc'}
\begin{aligned}
W_c'=W_{c}'(d):=\frac{\sqrt{\pi}}{\Gamma(1/4)2^{3/4}d}.
\end{aligned}
\end{equation}

Let \((\beta_i)_{i\in \Lambda_L}\sim \nu_{\Lambda_L}^{W,0}\). Then  for all \(0<W < W_{c}'\),
there are constants \(\kappa'=\kappa'(d,W)\),  and  \(C_{0}' (d,W)\) s.t. for any \(i,j\in \Lambda_L\),
\begin{equation}
\begin{aligned}
\mathbb{E}_{\Lambda_L}^{W,0}\left[ \mathcal{H}_{\beta,\Lambda_L}^{-1}(i,j)^{1/4} \right] \le  C_{0}' e^{-\kappa' |i-j|}.
\end{aligned}
\end{equation}
\end{theorem}


In this paper we will use the following corollary of the above result.

\begin{corollary}[decay with wired b.c. (2)]\label{cor:fractional}
Let  \(\Lambda_L\) be a finite box of side \(2L+1\) in  \(\mathbb{Z}^d\),
and let \((\beta_i)_{i\in \Lambda_L}\sim \nu_{\Lambda_L}^{W,\eta^{w}}\)
where \(\eta^w_{\Lambda_{L}}\) is the wired boundary condition introduced in
\eqref{eq-eta-boundary-Lambda}. Remember the definition of \(W_{c}'\) in \eqref{eq:Wc'}.

For all \(0<W < W_{c}',\) there are constants \(\kappa'=\kappa'(W,d)>0\) and \(C_{0}' (W,d)>0\) s.t. for any \(i,j\in \Lambda_L\),
\begin{equation}\label{decayd>1}
\begin{aligned}
\mathbb{E}_{\Lambda_L}^{W,\eta^{w}}\left[ \mathcal{H}_{\beta,\Lambda_L}^{-1}(i,j)^{1/4} \right]
\le  C_{0}' e^{-\kappa' |i-j|}.
\end{aligned}
\end{equation}

\end{corollary}

In an abuse of notation, in the rest of the paper we will write \(C_0\) for the constant in both decay results in Theorem \ref{thm-localization-of-the-ground-state-green-function}(ii) and Theorem \ref{thm-localization-fractional-moment}.
\medskip

\begin{proof}[Proof of Corollary \ref{cor:fractional}]
By \eqref{eq:marginal1} we have 
\[
\mathbb{E}_{\Lambda_{L}}^{W,\eta^{w}}\left[ \mathcal{H}_{\beta,\Lambda_L}^{-1}(i,j)^{1/4}
 \right]=
\mathbb{E}_{\Lambda_{L+1}}^{W,0}\left[ \mathcal{H}_{\beta,\Lambda_L}^{-1}(i,j)^{1/4} \right].
\]
By random walk representation (cf. Proposition 6 of \cite{Sabot2019}  and notations therein)
\[
\mathcal{H}_{\beta,\Lambda_L}^{-1}(i,j)=
\sum_{\sigma \in \mathcal{P}_{ij}^{\Lambda_{L} }}\frac{W_{\sigma }}{(2\beta )_{\sigma }}
\]
where \(\mathcal{P}_{ij}^{\Lambda_{L} }\) is the set of nearest neighbor paths from \(j_{0}\) to \(j\)
in \(\Lambda_{L}. \) 
It follows, for all \(i,j\in \Lambda_{L} ,\)
\begin{equation}\label{eq:boundLL+1}
\mathcal{H}_{\beta,\Lambda_{L}}^{-1}(i,j)
\leq \mathcal{H}_{\beta,\Lambda_{L+1}}^{-1}(i,j)
\end{equation}
since in the second term we have more paths. Therefore
\begin{align*}
&\mathbb{E}_{\Lambda_{L}}^{W,\eta^{w}}\left[ \mathcal{H}_{\beta,\Lambda_L}^{-1}(i,j)^{1/4}
 \right]=
\mathbb{E}_{\Lambda_{L+1}}^{W,0}\left[ \mathcal{H}_{\beta,\Lambda_L}^{-1}(i,j)^{1/4} \right]\\
&\qquad \leq \mathbb{E}_{\Lambda_{L+1}}^{W,0}\left[ \mathcal{H}_{\beta,\Lambda_{L+1}}^{-1}(i,j)^{1/4} \right]
\leq  C_{0}' e^{-\kappa' |i-j|},
\end{align*}
where in the last step we applied
Theorem \ref{thm-localization-fractional-moment}.
\end{proof}

\paragraph{Comparing  \(W_{c}\) and \(W_{c}'\).}
In Section 3.2 we will construct two sets of measure close to one \(\Omega_{1}\) (resp \(\Omega_{2}\))
using  Theorem \ref{cor:decayd=1} (resp. Corollary \ref{cor:fractional}).
Both sets can be used to construct the same lower bound on the IDS (cf Section 4), which will be valid
\(\forall \, W<W_{c}=W_{c} (d)\) (resp. \(\forall \, W<W_{c}'=W_{c}' (d)\)) if we use \(\Omega_{1}\) (resp. \(\Omega_{2}\).)

It is then reasonable to ask which result works for a larger set of parameters \(W,\) i.e. which of the two critical values is larger.
While we have an explicit numeric expression for \(W_{c}' (d),\) \(W_{c} (d)\) is only indirectly determined as the unique solution of \(F_{d} (W)=I_W e^{W(2d-2)}(2d-1)=1\) (cf. the Remark before Theorem \ref{cor:decayd=1}). The next lemma shows that  \(W_{c} (d) <W_{c}' (d)\)
for all $d\geq 2.$
\begin{lemma}\label{lem:comparaison}
For \(d=1\), \(W_{c} (1) =\infty>W_{c}' (1)\). For \(d\ge 2\), \(W_{c} (d) <W_{c}' (d)\).
\end{lemma}
\begin{proof}
Recalling the definition of modified Bessel function of the second kind
\(K_{\alpha }(x):=\int_0^{\infty}\cosh (\alpha t) e^{-x \cosh t} dt\), we have
\(I_W={2}e^W \sqrt{\frac{W}{2\pi}} K_0(W)\) and
\[
F_d(W)=\sqrt{\tfrac{{2}W}{\pi}}K_0(W) e^{W(2d-1)}(2d-1).
\]
Note that from \eqref{eq:Wc'} \(W_c'(d)=C/d\), where \(C\) is a constant independent of \(d,\) and hence
\(f (d):= F_{d} (W_{c}' (d))=F_{d} (C/d)\). Moreover
\begin{align*}
\partial_{W}F_{d} (W)&=\left (\frac{1}{2W}+ (2d-1)-\frac{{K_{1} (W)}}{K_{0} (W)}   \right)F_{d} (W)\\
\partial_{d}F_{d} (W)&=\left (2W+\frac{2}{2d-1}  \right)F_{d} (W).
\end{align*}
It follows
\begin{align*}
f' (d) &= -\tfrac{C}{d^{2}}\partial_{W}F_{d} (C/d)+ \partial_{d}F_{d} (C/d)\\
&=
F_{d} (C/d)
\left[\tfrac{2d+1}{2d (2d-1)} +\tfrac{C}{d^2}+\tfrac{C}{d^2}
\tfrac{K_1(C/d)}{K_{0} (C/d)}\right]>0,
\end{align*}
which proves that \(f\) is monotone increasing.

We compute numerically  \(F_2(W_c'(2))\approx 2.908\), hence \(W_c(2) < W_c'(2)\). As \(f\) is increasing, we have \(F_d(W_c'(d)) \ge F_2(W_c'(2)) \ge 2.9\), which implies \(W_{c}' (d)>W_{c} (d)\) \(\forall d\geq 2.\)
\end{proof}
\vspace{0,2cm}

Before going to the proof of the lower bound we list an additional useful corollary on the probability distribution of \(\mathcal{H}_{\beta ,\Lambda_{L} }^{-1} (0,0).\)

\begin{corollary}\label{corboundH00}
Let  \(\Lambda_L\) be a finite box of side \(2L+1\) in  \(\mathbb{Z}^d\),
and let \((\beta_i)_{i\in \Lambda_L}\sim \nu_{\Lambda_L}^{W,\eta^{w}}\)
where \(\eta^w_{\Lambda_{L}}\) is the wired boundary condition introduced in
\eqref{eq-eta-boundary-Lambda}.

Then, for any \(\delta >0\) we have
\begin{equation}\label{eq:bound H00}
\nu_{\Lambda_{L}}^{W,\eta^{w}}\left ( \mathcal{H}_{\beta,\Lambda_L}^{-1}(0,0)>\delta 
 \right)\leq \int_{0}^{\frac{1}{2\delta }} \frac{1}{\sqrt{\pi \gamma }} e^{-\gamma } d\gamma. 
\end{equation}
\end{corollary}
\begin{proof}
We argue
\[
\nu_{\Lambda_{L}}^{W,\eta^{w}}\left ( \mathcal{H}_{\beta,\Lambda_L}^{-1}(0,0)>\delta 
 \right)= \nu_{\Lambda_{L+1}}^{W,0}\left ( \mathcal{H}_{\beta,\Lambda_L}^{-1}(0,0)>\delta 
 \right).
\]
By \eqref{eq:boundLL+1} \( \mathcal{H}_{\beta,\Lambda_{L+1}}^{-1}(0,0)\geq  \mathcal{H}_{\beta,\Lambda_L}^{-1}(0,0)\) and hence 
\[
 \mathcal{H}_{\beta,\Lambda_L}^{-1}(0,0)>\delta \quad \Rightarrow \mathcal{H}_{\beta,\Lambda_{L+1}}^{-1}(0,0)> \mathcal{H}_{\beta,\Lambda_L}^{-1}(0,0)>
\delta.
\]
It follows
\begin{align*}
& \nu_{\Lambda_{L+1}}^{W,0}\left ( \mathcal{H}_{\beta,\Lambda_L}^{-1}(0,0)>\delta 
 \right)\leq  \nu_{\Lambda_{L+1}}^{W,0}\left ( \mathcal{H}_{\beta,\Lambda_{L+1}}^{-1}(0,0)>\delta 
 \right)\\
&= \nu_{\Lambda_{L+1}}^{W,0}\left (\frac{1}{2\mathcal{H}_{\beta,\Lambda_{L+1}}^{-1}(0,0)}<
\frac{1}{2\delta} 
 \right)= \int^{\frac{1}{2\delta }}_{0} \frac{1}{\sqrt{\pi \gamma }} e^{-\gamma } d\gamma ,
\end{align*}
where in the last step we used that \(1/[2\mathcal{H}_{\beta,\Lambda_{L+1}}^{-1}(0,0)]\)
is Gamma distributed (cf. Equation \eqref{eq-gamma-1over2G}).
\end{proof}

\section{Preliminary results}
\label{sec:orge0cd1d0}

\subsection{Connection between \(N (E)\) and the Green's function
with Dirichlet b.c.}

To obtain a lower bound on \(N (E, H_{\beta })\)   
we use the following classical argument (see, e.g. \cite{kirsch2007invitation})
\begin{lemma}\label{le1}
For any finite box \(\Lambda_L\) of side \(2L+1\) we have
\begin{equation}
N(E,{H}_{\beta }) \ge  \mathbb{E}^{W,\eta^{w}}_{\Lambda_{L}} [N_{\Lambda_L}(E,{H}^D_{\beta ,\Lambda_L})],
\end{equation}
where \(N_{\Lambda_L}(E,{H}^D_{\beta ,\Lambda_L})\) is the finite volume IDS with Dirichlet boundary condition  defined 
in \eqref{def:finite-vol-IDOS}, \( \mathbb{E}^{W,\eta^{w}}_{\Lambda_L}\) denotes the expectation with respect to the finite marginal \(\nu^{W,\eta^w}_{\Lambda_L}\) of \((\beta_{i})_{i\in \Lambda_{L}}\)
given in \eqref{eq-density-beta-Lambda}, and \(\eta^w=\eta^w_{\Lambda_L}\) are the wired boundary condition given in \eqref{eq-eta-boundary-Lambda}.
\end{lemma}

\begin{proof}
Recalling the definition of the integrated density of states \(N\), by \eqref{eq:Nviafinitevolume} we have
\begin{equation}
N (E)=\lim_{K\to \infty} \frac{1}{|\Lambda_K|} \mathbb{E}^{W,\eta^{w}_{\Lambda_K}}_{\Lambda_{K} }\left[\operatorname{tr} (\mathbbm{1}_{(-\infty,E]} ({H}_{\beta,\Lambda_{K}})) \right],
\end{equation}
where \(\Lambda_K\) is a finite box of side \(2K+1\). We split the large box \(\Lambda_K\) into a tiling of smaller boxes of side \(2L+1\) with \(L<K,\) \(\Lambda_{K}=\cup_{j=1}^{N_{K}}\Lambda_{L,j}.\) Using \((v_{i}-v_{j})^{2}\leq 2 (v_{i}^{2}+v_{j}^{2})\) (Dirichlet--Neumann bracketing) we obtain 
\[
{H}_{\beta,\Lambda_{K}}\leq \bigoplus_{j=1}^{N_{K}} {H}^{D}_{\beta,\Lambda_{L, j}},
\]
as a quadratic form. Note that by min-max principle, if \(A>B\) then \(\lambda_{A,j}>\lambda_{B,j}\), where \(\lambda_{A,j}\) are ordered eigenvalues of \(A\). This, together with translation invariance, the relation \(|\Lambda_{L}|N_{K}=|\Lambda_{K}|\) and \eqref{eq:marginal1} yields
\begin{equation}
\begin{aligned}
&\frac{1}{|\Lambda_K|} \mathbb{E}^{W,\eta^{w}_{\Lambda_K}}_{\Lambda_{K} }\left[\operatorname{tr} (\mathbbm{1}_{(-\infty,E]} (H_{\beta,\Lambda_{K}})) \right]
\geq \sum_{j=1}^{N_{K}}\frac{1}{|\Lambda_K|}   \mathbb{E}^{W,\eta^{w}_{\Lambda_K}}_{\Lambda_{K}}\left[\operatorname{tr} (\mathbbm{1}_{(-\infty,E]} (H^{D}_{\beta,\Lambda_{L, j}})) \right]\\
&=  \sum_{j=1}^{N_{K}} \frac{1}{|\Lambda_K|}  \mathbb{E}^{W,\eta^{w}_{\Lambda_{L,j}}}_{\Lambda_{L, j }}\left[\operatorname{tr} (\mathbbm{1}_{(-\infty,E]} (H^{D}_{\beta,\Lambda_{L, j}})) \right]\\
&= \frac{N_{K}}{|\Lambda_K|}  \mathbb{E}^{W,\eta^{w}_{\Lambda_L}}_{\Lambda_{L }}\left[\operatorname{tr} (\mathbbm{1}_{(-\infty,E]} (H^{D}_{\beta,\Lambda_{L}})) \right]=  \mathbb{E}^{W,\eta^{w}_{\Lambda_L}}_{\Lambda_{L}} [N_{\Lambda_L}(E,H^D_{\beta ,\Lambda_{L}})],
\end{aligned}
\end{equation}
for any finite box \(\Lambda_L\) in the family \(\{\Lambda_{L, j} \}_{j=1}^{N_{K}}\).  Taking the limit \(K\to \infty\),  keeping \(L\) fixed gives the desired result.
\end{proof}
\vspace{0,2cm}

As we are looking for a lower bound of \(N(E,H_{\beta})\), we can consider any finite box \(\Lambda_L\) (usually a larger \(L\) gives a better bound). We will fix  \(\Lambda_L=[-L,L]^d\cap \mathbb{Z}^d\) in the sequel. At the end we will choose \(L\) depending on the energy \(E.\)
\begin{lemma}\label{le2}
Let \(\Lambda_L=[-L,L]^d\cap \mathbb{Z}^d\). It holds
\begin{equation}
\mathbb{E}^{W,\eta^{w}}_{\Lambda_{L}} [N_{\Lambda_L}(E,H^D_{\beta ,\Lambda_{L}})]\geq
\frac{1}{|\Lambda_L|}
\nu^{W,\eta^{w}}_{\Lambda_L}\Big ((H^D_{\beta,\Lambda_{L}})^{-1}(0,0)\geq \frac{1}{E}\Big )
\end{equation}
\end{lemma}

\begin{proof}
\(H_{\beta,\Lambda_{L}}^{D}\) is a self adjoint finite random matrix, and by definition it is  a.s. positive definite. As a consequence, its smallest eigenvalue \(\lambda_1\) satisfies
\[
\lambda_{1}>0\quad \mbox{and}\qquad \frac{1}{\lambda_{1}}=\Vert (H^{D}_{\beta,\Lambda_{L}})^{-1}\Vert_{op},
\]
where \(\Vert \cdot \Vert_{op}\) stands for the operator norm.
It follows that
\begin{align}\label{eq-NE-bound-lambda1}
& \mathbb{E}^{W,\eta^{w}}_{\Lambda_{L}} [N_{\Lambda_L}(E,H^D_{\beta ,\Lambda_{L}})]= 
 \frac{1}{|\Lambda_L|} \mathbb{E}^{W,\eta^{w}}_{\Lambda_{L} }\left[\operatorname{tr} (\mathbbm{1}_{(-\infty,E]} (H^{D}_{\beta,\Lambda_{L}})) \right]\\
&\qquad \geq \frac{1}{|\Lambda_L|} \nu^{W,\eta^{w}}_{\Lambda_{L}}\Big ( \operatorname{tr} (\mathbbm{1}_{(-\infty,E]} (H^{D}_{\beta,\Lambda_{L}}) \ge 1\Big )\nonumber\\
&\qquad =\frac{1}{|\Lambda_L|} \nu^{W,\eta^{w}}_{\Lambda_{L}}\Big ( \lambda_{1} \leq E \Big )
=\frac{1}{|\Lambda_L|} \nu^{W,\eta^{w}}_{\Lambda_{L}}\Big ( \Vert (H^{D}_{\beta,\Lambda_{L}})^{-1}\Vert_{op}\geq \frac{1}{E} \Big ).\nonumber
\end{align}
Note that
\begin{equation}
\begin{aligned}
&\Vert (H^D_{\beta,\Lambda_{L}})^{-1}\Vert_{op}  =\sup_{\psi:\Vert \psi\Vert=1} \Vert (H^D_{\beta,\Lambda_{L}})^{-1}\psi\Vert \\
&\qquad \ge
\Vert (H^D_{\beta,\Lambda_{L}})^{-1} e_0\Vert  \geq  |(H^D_{\beta,\Lambda_{L}})^{-1}(0,0)|=(H^D_{\beta,\Lambda_{L}})^{-1}(0,0),
\end{aligned}
\end{equation}
where \(e_0=(\delta_{j0})_{j\in\mathbb Z^d}\). In the last step we used that, since the matrix is a.s. an M-matrix, the entries of its Green's function are all positive. Therefore,
\begin{equation}
\begin{aligned}
\mathbb{E}^{W,\eta^{w}}_{\Lambda_{L}} [N_{\Lambda_L}(E,H^D_{\beta ,\Lambda_{L}})] &\geq \frac{1}{|\Lambda_L|}
\nu^{W,\eta^{w}}_{\Lambda_{L}}\Big ( \Vert (H^{D}_{\beta,\Lambda_{L}})^{-1}\Vert_{op}\geq \frac{1}{E} \Big )\\
&\geq \frac{1}{|\Lambda_L|} \nu^{W,\eta^{w}}_{\Lambda_{L}}\Big ( (H^D_{\beta,\Lambda_{L}})^{-1}(0,0)\geq \frac{1}{E} \Big ).
\end{aligned}
\end{equation}
This concludes the proof of the lemma.
\end{proof}

\subsection{From Dirichlet  to simple boundary conditions. }

\begin{lemma}[Dirichlet versus simple bc (1)]
\label{corollary-localization-DS}
Let  \(\Lambda_{L}\) be the finite box  in \(\mathbb{Z}^d\) of side \(2L+1\) centered at \(0.\)
We consider  \((\beta_i)_{i\in \Lambda_L}\sim \nu_{\Lambda_L}^{W,\eta^{w}}\) where \(\eta^w\) is the wired boundary
condition introduced in \eqref{eq-eta-boundary-Lambda}.
Define \(\Omega_{1}=\Omega_{1,0}\cap \Omega_{1,1},\) with 
\begin{align}\label{Omega1}
\Omega_{1,0}:=\left\{ \sqrt{\tfrac{\mathcal{H}_{\beta,\Lambda_L}^{-1}(0,i)}{\mathcal{H}_{\beta,\Lambda_L}^{-1}(0,0)}} \le e^{-\kappa |i|/2},\ \forall i\in
\partial \Lambda_{L}\right\},\ 
\Omega_{1,1}:=\Big \{H_{\beta,\Lambda_L}^{-1}(0,0) \le e^{\kappa L}\Big\},
\end{align}
where \(\kappa\) is the constant introduced in Theorem \ref{cor:decayd=1} and remember that
 \(\mathcal{H}_{\beta,\Lambda_L}= WH_{\beta,\Lambda_{L}}\). Let \(W_{c}\) be as in
 Theorem \ref{thm-localization-of-the-ground-state-green-function}.

There is a constant
\(L_{0}=L_{0} (W,d) >1\) and \(C_{1}=C_{1} (d,W_{c})\) such that \(\forall\, L\geq L_{0}\) and
\(\forall \, 0<W<W_{c} \) we have  
\begin{equation}\label{eq:ome1prob}
\begin{aligned}
\nu_{\Lambda_L}^{W,\eta^{w}}\left (\Omega_{1,j}\right) \ge 1- C_{1}\  e^{-\kappa L/4}\qquad
\mbox{ for } j=0,1,
\end{aligned}
\end{equation}
and hence \(\nu_{\Lambda_L}^{W,\eta^{w}}\left (\Omega_{1}\right)\geq 1-2C_{1}\  e^{-\kappa L/4}.\)
Moreover on the set \(\Omega_{1}\), for any \(E>0\) it holds
\begin{equation}
\begin{aligned}
\left\{(H_{\beta,\Lambda_L}^S)^{-1}(0,0) > \frac{1}{E}\right\} \Rightarrow \left\{(H_{\beta,\Lambda_L}^D)^{-1}(0,0) > \frac{1}{2E} \right\},
\end{aligned}
\end{equation}
In particular, when \(d=1\) this result holds for all \(W>0,\) since \(W_{c}=\infty.\)  
\end{lemma}

Note that the set \(\Omega_{1,0}\) is measurable w.r.t. \(\{\beta_{j} \}_{j\in \Lambda\setminus \{0 \}}\)
while \(\Omega_{1,1}\) is measurable w.r.t. \((H_{\beta,\Lambda_L})^{-1}(0,0) .\)
This fact will be important in the proof of the lower bound for the IDS.\vspace{0,2cm}

\begin{proof}[Proof of Lemma \ref{corollary-localization-DS}]
 The decay estimate \eqref{eq-decayd=1} together with the Markov inequality entails that \(\forall \, 0<W<W_{c},\)
 \begin{align}
\nu_{\Lambda_L}^{W,\eta^{w}}\left (\Omega_{1,0}^{c}\right) &\leq
\sum_{i\in \partial \Lambda_{L}}\nu_{\Lambda_L}^{W,\eta^{w}}\left (\left (
\sqrt{\tfrac{\mathcal{H}_{\beta,\Lambda_L}^{-1}(0,i)}{\mathcal{H}_{\beta,\Lambda_L}^{-1}(0,0)}}
\right) > e^{-\kappa |i|/2}\right)\nonumber\\
&
\leq C_{0}|\partial\Lambda_{L}|e^{-\kappa L/2}\leq C_{1} e^{-\kappa L/4}
\end{align}
for some constants \(C_1\), and for \(L\) large enough depending on \(W,d\).
Corollary \ref{corboundH00} with \(\delta =e^{\kappa L/2}/W\) gives
\begin{align}\label{boundprob00}
\nu_{\Lambda_L}^{W,\eta^{w}}\left (\Omega_{1,1}^{c}\right) &=
\nu_{\Lambda_L}^{W,\eta^{w}}(H_{\beta,\Lambda_L}^{-1}(0,0) > e^{\kappa L })\\
&=
\nu_{\Lambda_L}^{W,\eta^{w}}(\mathcal{H}^{-1}_{\beta,\Lambda_L} (0,0)> e^{\kappa L }/W)
\leq C_{1} e^{-\kappa L/4}.\nonumber
\end{align}
Therefore \eqref{eq:ome1prob} holds.\vspace{0,2cm}

Assume now we are in \(\Omega_{1}.\)  By resolvent identity we have
\begin{align}\label{resolv}
&(H_{\beta,\Lambda_L}^S)^{-1}(0,0)- (H_{\beta,\Lambda_L}^D)^{-1}(0,0)=\\
&\qquad 
 \sum_{j\in \partial \Lambda_{L}} (H_{\beta,\Lambda_L}^S)^{-1}(0,j) (2d-n_{j})
 (H_{\beta,\Lambda_L}^D)^{-1}(j,0) .\nonumber
\end{align}
By random walk representation we have, setting
 \(\mathcal{P}_{j_{0},j}^{\Lambda_{L} }=\)  the set of nearest neighbor paths from \(j_{0}\) to \(j\)
in \(\Lambda_{L}, \)  (cf. Proposition 6 of \cite{Sabot2019}  and notations therein)
\begin{equation}\label{eq-S-1-larger-D-1}
\begin{aligned}
\mathcal{H}_{\beta,\Lambda_L}^{-1}(0,i) = \sum_{\sigma\in \mathcal{P}^{\Lambda_L}_{0,i}}\frac{W_{\sigma}}{(2 \beta)_{\sigma}} \ge \sum_{\sigma\in \mathcal{P}^{\Lambda_L}_{0,i}}\frac{W_{\sigma}}{(2 \beta+W(2d-n))_{\sigma}}
= (\mathcal{H}_{\beta,\Lambda_L}^D)^{-1}(0,i).
\end{aligned}
\end{equation}
Therefore, on the set \(\Omega_{1}\) we have
\begin{align*}
&\left| \frac{(H_{\beta,\Lambda_L}^D)^{-1}(0,0)}{(H_{\beta,\Lambda_L}^S)^{-1}(0,0)} -1\right|=
1-\frac{(H_{\beta,\Lambda_L}^D)^{-1}(0,0)}{(H_{\beta,\Lambda_L}^S)^{-1}(0,0)} \\
&\qquad =
\sum_{j\in \partial \Lambda_{L}} (H_{\beta,\Lambda_L}^S)^{-1}(0,j) (2d-n_{j})
 \frac{(H_{\beta,\Lambda_L}^D)^{-1}(j,0)}{( H_{\beta,\Lambda_L}^S)^{-1}(0,0)}\\
&\qquad \leq   (H_{\beta,\Lambda_L}^S)^{-1}(0,0) \sum_{j\in \partial \Lambda_{L}}
\frac{(\mathcal{H}_{\beta,\Lambda_L}^S)^{-1}(0,j)}{(\mathcal{H}_{\beta,\Lambda_L}^S)^{-1}(0,0)} (2d-n_{j})
 \frac{(\mathcal{H}_{\beta,\Lambda_L}^S)^{-1}(j,0)}{(\mathcal{H}_{\beta,\Lambda_L}^S)^{-1}(0,0)}\\
&\qquad \leq e^{\kappa L} 2d\sum_{j\in \partial \Lambda_{L}} e^{-2\kappa |j|}\leq 2d
|\partial \Lambda_{L}| e^{+\kappa L}e^{- 2\kappa L}\leq  e^{- \kappa L/4}\leq \frac{1}{2},
\end{align*}
for \(L\geq L_{0}\), where \(L_0\) depends on \(W,d\). Thus, on \(\Omega_{1},\)  it holds
\begin{align}\label{boundprob00}
({H}_{\beta,\Lambda}^{D})^{-1} (0,0)&=({H}_{\beta,\Lambda}^{S})^{-1} (0,0)
\left[1- \left(1-\frac{(H_{\beta,\Lambda_L}^D)^{-1}(0,0)}{(H_{\beta,\Lambda_L}^S)^{-1}(0,0)} \right) \right] \\
& \geq
\frac{1}{2}({H}_{\beta,\Lambda}^{S})^{-1} (0,0),\nonumber
\end{align}
and hence
\begin{align*}
( H_{\beta,\Lambda_L}^S)^{-1}(0,0) > 1/E\ \Rightarrow\ 
({H}_{\beta,\Lambda}^{D})^{-1} (0,0)\geq \frac{1}{2E}.
\end{align*}
\end{proof}\vspace{0,2cm}

\begin{lemma}[Dirichlet versus simple bc (2)]
\label{corollary-fractional-moment}
Let  \(\Lambda_{L}\) be the finite box  in \(\mathbb{Z}^d\) of side \(2L+1\) centered at \(0.\)
We consider  \((\beta_i)_{i\in \Lambda_L}\sim \nu_{\Lambda_L}^{W,\eta^{w}}\) where \(\eta^w\) is the wired boundary
condition introduced in \eqref{eq-eta-boundary-Lambda}.

Define \(\Omega_{2}=\Omega_{2,0}\cap \Omega_{2,1},\) with 
\begin{align}\label{Omega2}
\Omega_{2,0}:=\Big \{\max_{j\in \partial \Lambda_L,i\sim 0}
\left( \mathcal{H}_{\beta,\Lambda_L\setminus\{0 \}}^{-1}(i,j)\right)
 \le  e^{-\frac{3}{2}\kappa' L}\Big \},\ 
\Omega_{2,1}:=\Big \{H_{\beta,\Lambda_L}^{-1}(0,0) \le e^{\kappa' L}\Big\},
\end{align}
where \(\kappa'\) is the constant introduced in Theorem \ref{thm-localization-fractional-moment}
and remember that
 \(\mathcal{H}_{\beta,\Lambda_L}= WH_{\beta,\Lambda_{L}}\).  Let \(W_{c}'\) be as in
 Theorem \ref{thm-localization-fractional-moment}.

 There is  \(L_{0} (W,d)>1\) such that \(\forall \, L\geq L_{0}\) and \(\forall \, 0<W<W_{c}' \) we have \(\kappa' L>1,\) and there is
a constant \(C_{1}'=C_{1}' (W,d) >0\) such that  
\begin{equation}\label{probomega2}
\begin{aligned}
\nu_{\Lambda_L}^{W,\eta^{w}}\left (\Omega_{2,j}\right) \ge 1- C_{1}'\  e^{-\kappa' L/4}\qquad
\mbox{ for } j=0,1,
\end{aligned}
\end{equation}
and hence \(\nu_{\Lambda_L}^{W,\eta^{w}}\left (\Omega_{2}\right)\geq 1-2C_{1}'\  e^{-\kappa' L/4}.\)
Moreover on the set \(\Omega_{2}\) it holds
\begin{equation}
\begin{aligned}
\left\{(H_{\beta,\Lambda_L}^S)^{-1}(0,0) > \frac{1}{E}\right\} \Rightarrow \left\{(H_{\beta,\Lambda_L}^D)^{-1}(0,0) > \frac{1}{2E} \right\},
\end{aligned}
\end{equation}
for all energy \(0<E<\frac{1}{2}.\)
\end{lemma}

Note that also here the set \(\Omega_{2,0}\) is measurable w.r.t. \(\{\beta_{j} \}_{j\in \Lambda\setminus \{0 \}}\)
while \(\Omega_{2,1}\) is measurable w.r.t. \((H_{\beta,\Lambda_L})^{-1}(0,0) .\)
This fact will be important in the proof of the lower bound on the IDS. \vspace{0,2cm}

\begin{proof}[Proof of Lemma \ref{corollary-fractional-moment}]
Using the random path representation, as in \eqref{eq:boundLL+1}, we obtain
\(\mathcal{H}_{\beta,\Lambda_L\setminus\{0\} }^{-1}(i,j)\leq  \mathcal{H}_{\beta,\Lambda_L }^{-1}(i,j).\)
Then, the decay estimate \eqref{decayd>1} together with the Markov inequality entails,  \(\forall \, 0<W<W_{c}',\)
\begin{align*}
&\nu_{\Lambda_L}^{W,\eta^{w}}\left (\Omega_{2,0}^{c}\right) \leq
\sum_{j\in \partial \Lambda_{L}, i\sim 0}\nu_{\Lambda_L}^{W,\eta^{w}}\left (
\mathcal{H}_{\beta,\Lambda_L\setminus\{0\} }^{-1}(i,j)> e^{-\frac{3}{2}\kappa' L}\right)\\
&\qquad \leq \sum_{j\in \partial \Lambda_{L}, i\sim 0}
e^{\kappa' L 3/8}\mathbb{E}_{\Lambda_L}^{W,\eta^{w}}
\left[(\mathcal{H}_{\beta,\Lambda_L\setminus\{0\}}^{-1}(i,j))^{1/4} \right]\\
&\qquad \leq \sum_{j\in \partial \Lambda_{L}, i\sim 0}
e^{\kappa' L 3/8}\mathbb{E}_{\Lambda_L}^{W,\eta^{w}}
\left[(\mathcal{H}_{\beta,\Lambda_L}^{-1}(i,j))^{1/4} \right]\\
&\qquad 
\leq C_{0} \sum_{j\in \partial \Lambda_{L}, i\sim 0} e^{\kappa' L 3/8} e^{-\kappa' |i-j|}\leq C_{0}|\partial \Lambda_{L}| e^{-\kappa' L 5 /8}\leq C_{1}' e^{-\kappa' L/4}
\end{align*}
for some constant \(C_1'=C_1'(W,d)\). 
The bound for \(\Omega_{2,1}\) works exactly as the one for \(\Omega_{1,1}\)
in Lemma \ref{corollary-localization-DS}.
Therefore  \eqref{probomega2} holds.\vspace{0,2cm}

Assume now we are on \(\Omega_2\). We have, for all \(j\in \partial \Lambda_{L}\)
\begin{equation}\label{rel0j}
\mathcal{H}_{\beta,\Lambda_L}^{-1}(0,j)= 
\sum_{i\sim 0} \mathcal{H}_{\beta,\Lambda_L}^{-1}(0,0) W 
\mathcal{H}_{\beta,\Lambda_L\setminus\{0\}}^{-1}(i,j).
\end{equation}
Therefore
\[
\mathcal{H}_{\beta,\Lambda_L}^{-1}(0,j)\leq W2d
e^{\kappa' L} e^{-\frac{3}{2}\kappa' L }
\leq W e^{-\frac{1}{4}\kappa' L }.
\]
By \eqref{resolv} and \eqref{eq-S-1-larger-D-1}
\begin{equation}
\begin{aligned}
&\hspace{-2cm}|(H_{\beta,\Lambda_L}^D)^{-1}(0,0)-(H_{\beta,\Lambda_L}^S)^{-1}(0,0)|
  \\
  &= \sum_{j\in \partial \Lambda_{L}} (H_{\beta,\Lambda_L}^S)^{-1}(0,j) (2d-n_{j})
 (H_{\beta,\Lambda_L}^D)^{-1}(j,0) \\
&\leq \sum_{j\in \partial \Lambda_{L}} (H_{\beta,\Lambda_L}^S)^{-1}(0,j) (2d-n_{j})
 (H_{\beta,\Lambda_L}^S)^{-1}(j,0)
  \\
  &\leq 
    W^{2}| \partial \Lambda_L| 2d  e^{-\frac{1}{2}\kappa' L}\\
  &\leq  e^{-\frac{1}{4}\kappa' L}
\end{aligned}
\end{equation}
for \(L\) large enough, depending on \(W,d\). It follows that, if \(E \le 1/2\)
\begin{align*}
&(H_{\beta,\Lambda_L}^S)^{-1}(0,0)> \frac{1}{E} \Rightarrow \\
&\qquad (H_{\beta,\Lambda_L}^D)^{-1}(0,0) \ge (H_{\beta,\Lambda_L}^S)^{-1}(0,0)-e^{-\frac{1}{4}\kappa' L}\ge \frac{1}{E}- e^{-\frac{1}{4}\kappa' L} \ge \frac{1}{2 E}
\end{align*}
for \(L\geq L_0=L_0(W,d)\).
\end{proof}

\section{Lower bound on the IDS}
\label{sec:org8c54fb3}

We are now ready to prove Theorem \ref{thm-main-theorem}.
By  Lemma \ref{le1} and Lemma \ref{le2} in Section 3.1 we have
\[
N(E,H_{\beta })\geq \frac{1}{|\Lambda_{L}|} \nu^{W,\eta^{w} }_{\Lambda_{L}}\left (( H^{D}_{\beta ,\Lambda_{L}})^{-1} (0,0) \geq \frac{1}{E}
\right ).
\]
Remember the definition of \(W_{cr}\) in \eqref{def:Wcrit} and the configuration sets \(\Omega_{1,0},\Omega_{1,1},\Omega_{2,0},\Omega_{2,1}\)
introduced in \eqref{Omega1} and
\eqref{Omega2}. We define, for \(j=0,1,\)  
\begin{equation}\label{def:omegaloc}
\Omega_{\text{loc},j}:= \left\{
\begin{array}{ll}
\Omega_{1,j} & \mbox{ if } W_{cr}=W_{c}, \\
\Omega_{2,j}& \mbox{ if } W_{cr}=W_{c}'. \\
\end{array}
\right.
\end{equation}
and \(\Omega_{\text{loc}}:= \Omega_{\text{loc},0}\cap \Omega_{\text{loc},1}.\)
Note that the constants  $\kappa,C_{0},C_{1}, \kappa',C_{0}', C_{1}'$ appearing in the definition of the sets  $\Omega_{1,j}, \Omega_{2,j}$ and the results of   Lemma
\ref{corollary-localization-DS} and \ref{corollary-fractional-moment} 
play the same role.  To alleviate notation, in the following we will not distinguish them.

By Lemma \ref{corollary-localization-DS} and \ref{corollary-fractional-moment}  we have
\begin{equation}\label{eq:lb1}
N (E,H_{\beta })\geq \frac{1}{|\Lambda_{L}|} \mathbb{E}^{W,\eta^{w} }_{\Lambda_{L}}\left [
\mathbf{1}_{\Omega_{\text{loc}}} \mathbf{1}_{\left\{  H_{\beta ,\Lambda_{L}}^{-1} (0,0) \geq \frac{1}{2E}  \right\}}
\right ]
\end{equation}
for all \(0<W<W_{cr},\) \(E<\frac{1}{2}\) and \(L\) large.
Remember that \(H_{\beta,\Lambda_{L}}=\frac{1}{W}\mathcal{H}_{\beta,\Lambda_{L}},\) and
set  \(0^{c}=\Lambda_{L} \setminus \{0 \}.\)
 By Schur decomposition,
\begin{equation}
\begin{aligned}
\frac{1}{W}H_{\beta,\Lambda_L}^{-1}(0,0)= \mathcal{H}_{\beta,\Lambda_L}^{-1}(0,0)=
\frac{1}{2\beta_{0}- P_{0,0^c}^{W} \mathcal{H}^{-1}_{\beta ,\Lambda_{L} \setminus \{0 \}} P^{W}_{0^c,0}}=:\frac{1}{y}.
\end{aligned}
\end{equation}

\begin{lemma}The  conditional density of the variable \(y:=2\beta_{0}- P_{0,0^c}^{W} \mathcal{H}^{-1}_{\beta ,\Lambda_{L} \setminus \{0 \}} P^{W}_{0^c,0},\)
 given \(\beta_{0^{c}}= (\beta_{j})_{i\in 0^{c}}\), is 
\begin{equation}\label{conditional-density-y}
d\rho_{a_{0}} (y)=\rho_{a_{0}} (y) dy=\frac{e^{a_{0}}}{\sqrt{2\pi }} e^{-\frac{1}{2}\left(y+\frac{a_{0}^2}{y} \right)}\frac{1}{\sqrt{y}} \mathbf{1}_{y>0} dy,
\end{equation}
where
\begin{equation}\label{a-def}
a_{0}=a_{0} (\beta_{0^{c}}) =\sum_{j\in \partial \Lambda_{L}}\frac{\mathcal{H}_{\beta,\Lambda_{L}}^{-1} (0,j)}{
\mathcal{H}_{\beta,\Lambda_{L}}^{-1} (0,0)}\eta^{w}_{j}=
 W \hspace{-0,4cm}\sum_{i\sim 0, j\in \partial \Lambda_{L}}  \hspace{-0,4cm}
\mathcal{H}_{\beta,\Lambda_{L}\setminus \{0 \}}^{-1}(i,j)\ \eta^{w}_{j} .
\end{equation}
\end{lemma}
The corresponding average will be denoted by \(\mathbb{E}^{a_{0}}.\)\medskip

\begin{proof}
The result follows from the factorization in \cite[Equation (5.14)]{Sabot2019} with \(U=\Lambda \setminus \{0\}\) and \(U^c=\{0\}\).

Alternatively one may insert  the following relations in \eqref{eq-density-beta-Lambda}:
\begin{align}\label{}
 \left< 1,\mathcal{H}_{\beta,\Lambda_{L}}1 \right> & = 2\beta_{0}+ \left< 1_{0^{c}},\mathcal{H}_{\beta,\Lambda_{L}\setminus \{0 \}}1_{0^{c}} \right>-2\sum_{i\sim 0}W=y+F (\beta_{0^{c}})\\
 \left< \eta^{w}, \mathcal{H}_{\beta,\Lambda_L}^{-1} \eta^{w} \right>&=
\left< \eta^{w}, \mathcal{H}_{\beta,\Lambda_L\setminus \{0 \}}^{-1} \eta^{w} \right> + a^{2}_{0} \mathcal{H}_{\beta,\Lambda_L}^{-1} (0,0)=
 G (\beta_{0^{c}}) + \frac{a^{2}_{0}}{y}\\
\det \mathcal{H}_{\beta,\Lambda_L}&= y \det \mathcal{H}_{\beta,\Lambda_{L}\setminus \{0 \}} ,
\end{align}
where, \(F,G\) are functions of \(\beta_{0^{c}}\) and  in second line we combined \(\eta^{w}_{0} =0\) and 
the resolvent identity
\(A^{-1}=B^{-1}+B^{-1} (B-A)B^{-1}+ B^{-1} (B-A)A^{-1} (B-A)B^{-1}\) with
\(A=\mathcal{H}_{\beta,\Lambda_{L}}\) and
\(B=2\beta_{0}\oplus  \mathcal{H}_{\beta,\Lambda_{L}\setminus \{0 \}}.\)
This yields
\[
\rho_{a_{0}} (y)=c_{a_{0}} e^{-\frac{1}{2}\left(y+\frac{a_{0}^2}{y} \right)}\frac{1}{\sqrt{y}} \mathbf{1}_{y>0}.
\]
To determine the normalizing constant $c_{a_{0}}$ note that
\[
\int_{0}^{\infty}  e^{-\frac{1}{2}\left(y+\frac{a_{0}^2}{y} \right)}\frac{1}{\sqrt{y}} dy =\frac{  \sqrt{2\pi } }{a_{0}} e^{-a_{0}}\  \mathbb{E}_{IG (a_{0},a_{0}^{2})}[y]=
 \sqrt{2\pi } e^{-a_{0}},
\]
where $\mathbb{E}_{IG (a_{0},a_{0}^{2})}$ denotes the expectation with respect to the probability distribution of the
inverse Gaussian $IG (a_{0},a_{0}^{2}). $ 

\end{proof}\medskip

Using these results, the average in  \eqref{eq:lb1} can be reformulated as follows.
\begin{align}
& \mathbb{E}^{W,\eta^{w} }_{\Lambda_{L}}\left [
\mathbf{1}_{\Omega_{\text{loc}}} \mathbf{1}_{\left\{  H_{\beta ,\Lambda_{L}}^{-1} (0,0) \geq \frac{1}{2E}  \right\}}
\right ]= \mathbb{E}^{W,\eta^{w} }_{\Lambda_{L}}\left [
\mathbf{1}_{\Omega_{\text{loc}}} \mathbf{1}_{\left\{  \frac{1}{W}H_{\beta ,\Lambda_{L}}^{-1} (0,0) \geq \frac{1}{2EW}  \right\}}
\right ]\\
&\quad 
 = \mathbb{E}^{W,\eta^{w} }_{\Lambda_{L}}\left [\mathbf{1}_{\Omega_{\text{loc},0}}
\int \mathbf{1}_{\Omega_{\text{loc},1}}  \mathbf{1}_{y\leq 2EW} \frac{e^{a_{0}}}{\sqrt{2\pi }} e^{-\frac{1}{2}\left(y+\frac{a^{2}_{0}}{y} \right)}\frac{1}{\sqrt{y}} \mathbf{1}_{y>0} dy\right ]\nonumber\\
&\quad 
 =\mathbb{E}^{W,\eta^{w} }_{\Lambda_{L}}\left [ \mathbf{1}_{\Omega_{\text{loc},0}} \mathbb{E}^{a_{0} (\beta_{0^{c}})}[ \mathbf{1}_{\Omega_{\text{loc},1}} \mathbf{1}_{y\leq 2EW} ]  \right ]\\
 &\quad=  \mathbb{E}^{W,\eta^{w} }_{\Lambda_{L}}\left [ \mathbf{1}_{\Omega_{\text{loc},0}}\ 
 \mathbb{E}^{a_{0} (\beta_{0^{c}})} [\mathbf{1}_{ We^{-\kappa L}\leq y\leq 2EW }]\ \right ],
\nonumber
\end{align}
where we used that \(\Omega_{\text{loc},0}\) is measurable w.r.t. \(\beta_{0^{c}}\) and
\[
\Omega_{\text{loc},1}= \Big \{H_{\beta,\Lambda_L}^{-1}(0,0) < e^{\kappa L}\Big\}=
\big \{ \frac{y}{W} > e^{-\kappa L}\big\}.
\]
We have for \(L\) large enough, \(\forall \beta_{0^{c}}\in \Omega_{1,0}\),
\begin{align*}
a_{0} (\beta_{0^{c}})&=\sum_{j\in \partial \Lambda_{L}}\frac{\mathcal{H}_{\beta,\Lambda_{L}}^{-1} (0,j)}{
                       \mathcal{H}_{\beta,\Lambda_{L}}^{-1} (0,0)}\eta^{w}_{j}\leq  2 dW|\partial \Lambda_{L}| e^{-\kappa L}\leq W  e^{-\kappa L/2}
\end{align*}
and \(  \forall \beta_{0^{c}}\in \Omega_{2,0}\),
\begin{align*}
a_{0} (\beta_{0^{c}})&=W \sum_{i\sim 0, j\in \partial \Lambda_{L}}\mathcal{H}_{\beta,\Lambda_{L}\setminus \{0 \}}^{-1} (i,j)
\eta_{j}\leq 2dW|\partial \Lambda_{L}| e^{-\frac{3}{2}\kappa L}\leq W e^{-\kappa L/2}.
\end{align*}
Setting \(\overline{a}_{0}:=  W e^{-\kappa L/2},\) and remarking that \(a_{0}\geq 0,\) we argue
\[
e^{a_{0} (\beta_{0^{c}})}e^{-\frac{(a_{0}(\beta_{0^{c}}))^{2}}{2y}}\geq e^{-\frac{(a_{0}(\beta_{0^{c}}))^{2}}{2y}}\geq e^{-\frac{\overline{a}^{2}_{0}}{2y}}
= e^{\overline{a}_{0}}e^{-\frac{\overline{a}^{2}_{0}}{2y}} e^{-\overline{a}_{0}},
\]
hence
\[
\rho_{a_{0} (\beta_{0^{c}})  } (y) \geq \rho_{\overline{a}_{0}} (y) e^{-\overline{a}_{0}}\qquad  \forall \beta_{0^{c}}\in \Omega_{\text{loc},1}.
\]

Therefore we obtain
\begin{align*}
&\mathbb{E}^{W,\eta^{w} }_{\Lambda_{L}}\left [ \mathbf{1}_{\Omega_{\text{loc},0}}\ 
 \mathbb{E}^{a_{0} (\beta_{0^{c}})} [\mathbf{1}_{ We^{-\kappa L}\leq y\leq 2WE }]\ \right ]\\
&\qquad 
\geq \ e^{-\overline{a}_{0}}\ \nu^{W,\eta^{w} }_{\Lambda_{L}} \big (\Omega_{\text{loc},0} \big)
\ \rho_{\overline{a}_{0}}  \big(W e^{-\kappa L}\leq y\leq 2WE   \big)\\
&\qquad \geq  (1-C_{2}e^{-\kappa L/4})\
\rho_{\overline{a}_{0}}  \big( W e^{-\kappa L}\leq y\leq 2WE   \big)
\end{align*} for some constant \(C_2\),
where we used \eqref{eq:ome1prob}, \eqref{probomega2} and the bound
\( e^{-\overline{a}_{0}}=  e^{-We^{-\frac{1}{2}\kappa L}}\geq (1-c e^{-\kappa L/4})\)
for some constant \(c>0.\)

It remains to extract a lower bound on
\(\rho_{\overline{a}_{0}}  \big( We^{-\kappa L}\leq y\leq 2WE   \big).\) Set
\(L=L_{E}:= \frac{1}{\kappa }\ln \frac{1}{E}.\) For \(E\) small,  \(L_{E}\) is large enough for all our results to hold.
Then \(2WE=2We^{-\kappa L}\) and hence \(We^{-\kappa L}< 2WE<\overline{a}_{0}.\) Moreover 
\[
\frac{(\overline{a}_{0}-y)^{2}}{y}\leq W\qquad \forall y \text{ such that } \, We^{-\kappa L}\leq y\leq 2 W e^{-\kappa L}.
\]
It follows
\begin{align*}
&\rho_{\overline{a}_{0}}  \big( We^{-\kappa L}\leq y\leq 2WE\big)
=\frac{1}{\sqrt{2\pi }} \int_{We^{-\kappa L}}^{2WE}
e^{-\frac{(\overline{a}_{0}-y)^{2}}{2y}}\frac{1}{\sqrt{y}} dy\\
&
\geq \ \frac{e^{-\frac{W}{2}}}{\sqrt{2\pi }} \int_{We^{-\kappa L}}^{2WE}\frac{1}{\sqrt{y}} dy
= c e^{-W/2}\sqrt{WE}
\end{align*}
for some constant \(c>0\) independent of \(E,\kappa,L, W\).
Putting all these results together we obtain
\begin{equation}
\begin{aligned}
N(E,H_{\beta }) &\ge  \tfrac{1}{|\Lambda_{L_{E}}|}
 \nu^{W,\eta^{w} }_{\Lambda_{L_{E}}}\left (( H^{D}_{\beta ,\Lambda_{L_{E}}})^{-1} (0,0) \geq \frac{1}{E}
\right)\\
& \geq   \tfrac{1}{|\Lambda_{L_{E}}|}(1-C_{2}e^{-\kappa L_{E}/4})\
 ce^{-W/2}\sqrt{W} \sqrt{E}\ \ge\  c'\frac{1}{|\log E|^d} \sqrt{E}
\end{aligned}
\end{equation}
for some constant \(c'>0\) depending on \(W,d,\kappa\).
This concludes the proof of the lower bound.

\section{Wegner estimate}
\label{sec:org59d4fb1}

In this section we prove Theorem \ref{wegner-type-theorem}. We work out in detail 
the proof only for the case of simple boundary conditions. The proof in the case of Dirichlet boundary conditions works exactly in the same way. 

\begin{proof}[Proof of Theorem \ref{wegner-type-theorem}]
Note that, 
since \(H=\mathcal{H}/W,\) we have
\[
N_{\Lambda_L}(E,{H}_{\beta ,\Lambda_{L}})=N_{\Lambda_L}(WE,\mathcal{H}_{\beta ,\Lambda_{L}}) =\frac{1}{|\Lambda_L|}
\operatorname{tr} \mathbbm{1}_{(-\infty,0]}  (\mathcal{H}_{\beta ,\Lambda_{L}}-WE ).
\]
We start with a regularity bound on
\[N_{\Lambda_L}(E+\varepsilon,{H}_{\beta ,\Lambda_{L}})-N_{\Lambda_L}(E-\varepsilon,{H}_{\beta ,\Lambda_{L}})=\frac{1}{|\Lambda_L|}
\operatorname{tr} \mathbbm{1}_{[-\varepsilon,\varepsilon]}  (\mathcal{H}_{\beta ,\Lambda_{L}}-WE ).\]
We smooth out the discontinuous function \( \mathbbm{1}_{[-\varepsilon,\varepsilon]} (x)\) as follows.
Let \(\rho \) be a  smooth non-decreasing function \(\rho\)  satisfying  \(\rho=0\) on
\((-\infty,-\varepsilon )\) and \(\rho=1\) on \((\varepsilon ,\infty)\).
Then
\[
\mathbbm{1}_{[-\varepsilon ,\varepsilon ]} (x)\leq \rho (x+2\varepsilon )-\rho (x-2\varepsilon )\qquad \forall x\in \mathbb{R}.  
\]
Setting
\begin{equation}\label{def:deltaL}
\delta_L(E,\varepsilon):=\mathbb{E}_{\Lambda_{L}}^{W,\eta^{w}}\big [ \operatorname{tr}\left ( \rho(\mathcal{H}_{\beta ,\Lambda_{L}} -E+2 \varepsilon)-\rho(\mathcal{H}_{\beta,\Lambda_L}-E-2\varepsilon) \right)\big ],
\end{equation}
we have
\begin{equation}
0\leq \mathbb{E}_{\Lambda_{L}}^{W,\eta^{w}}\left[ N_{\Lambda_L}(E+\varepsilon,H_{\beta,\Lambda_L})-N_{\Lambda_L}(E-\varepsilon,H_{\beta ,\Lambda_L}) \right] \le \frac{1}{|\Lambda_L|} \delta_L(WE,W\varepsilon).
\end{equation}
Remember that   the conditional
measure for \(2\beta_{j}- P_{j,j^c}^{W} \mathcal{H}^{-1}_{\beta ,\Lambda_{L} \setminus \{j \}} P^{W}_{j^c,j},\) 
 given \(\beta_{j^{c}}= (\beta_{j})_{i\in \Lambda_{L} \setminus \{j \}},\) is \(\rho_{a_{j}}\)
defined in \eqref{conditional-density-y}, where, for general \(j\in \Lambda_{L},\)
\begin{equation}\label{eq:adefj}
a_{j} =a_{j} (\beta_{j^{c}}) := \eta^{w}_{j}+W \hspace{-0,4cm}\sum_{i\sim j, k\in  \Lambda_{L} \setminus \{j \}}  \hspace{-0,4cm}
\mathcal{H}_{\beta,\Lambda_{L} \setminus \{j \}}^{-1}(i,k)\ \eta^{w}_{k}
=\frac{\sum_{k\in \partial \Lambda_{L}}\mathcal{H}_{\beta,\Lambda_{L}}^{-1} (j,k)\eta^{w}_{k}}{
\mathcal{H}_{\beta,\Lambda_{L}}^{-1} (j,j)}
\end{equation}
(cf. Equation (5.14) of \cite{Sabot2019}).
We also recall that the L\'evy concentration of a measure \(\mu\) on \(\mathbb{R}\) is
defined by
\begin{equation}\label{levy-concentrtion}
\mathscr{L}_{\mu}(\varepsilon)=\sup_{x}\mu([x,x+\varepsilon)).
\end{equation}
By Lemma \ref{le:deltabound} below, we have
\begin{align*}
&\mathbb{E}_{\Lambda_{L}}^{W,\eta^{w}}\left[ N_{\Lambda_L}(E+\varepsilon,H_{\beta,\Lambda_L})-N_{\Lambda_L}(E-\varepsilon,H_{\beta ,\Lambda_L}) \right] \leq \frac{1}{|\Lambda_{L}|} \delta_L(WE,W\varepsilon)\\
&\qquad 
\leq\frac{1}{|\Lambda_{L}|} \sum_{j\in \Lambda_{L}}\mathbb{E}_{\Lambda_{L}}^{W,\eta^{w}}\left[
\mathscr{L}_{\rho_{a_{j} (\beta_{j^{c}} )}}(4W\varepsilon) \right].
\end{align*}
Now \eqref{eq:wegner-17} and \eqref{eq:wegner-18} follow by inserting the bounds \eqref{eq:levy-bound} and \eqref{eq:levy-bound-bis}, stated below. This concludes the proof of Theorem \ref{wegner-type-theorem}.
\end{proof}

\begin{lemma}\label{le:deltabound}
For all \(E>0\) and \(0<\varepsilon<E \) it holds
\begin{equation}\label{eq:deltabound}
\delta_L(E,\varepsilon)\leq  \sum_{j\in \Lambda_L }
\mathbb{E}_{\Lambda_{L}}^{W,\eta^{w}}\left[
\mathscr{L}_{\rho_{a_{j} (\beta_{j^{c}})}}(4\varepsilon)  \right].
\end{equation}
\end{lemma}

\begin{proof}
Note  that 
\[
\mathcal{H}_{\beta,\Lambda_{L}}\pm 2\varepsilon = 2 (\beta\pm \varepsilon  )-P^{W}=
\mathcal{H}_{\beta\pm \varepsilon ,\Lambda_{L}}.
\]
Order the vertices in \(\Lambda_L\) as \(\{1,2 ,\ldots,|\Lambda_{L}|\}\). For each
\(1\le k\le |\Lambda_L|\), we define
\begin{equation}\label{def:omegatilde}
{\beta}'_k={\beta}'_k(\varepsilon):= (\beta_{1} +\varepsilon, ,\ldots,\beta_{k-1}+\varepsilon,\beta_k-\varepsilon ,\ldots,\beta_{|\Lambda_L|}-\varepsilon)
\end{equation}
and
\[{\beta}'_{|\Lambda_L|+1}={\beta}'_{|\Lambda_L|+1}(\varepsilon):= (\beta_{1} +\varepsilon, ,\ldots,\beta_{|\Lambda_L|}+\varepsilon).\]
With this convention we have
\[
\mathcal{H}_{{\beta}'_{k+1},\Lambda_{L}} (i,j) =\mathcal{H}_{{\beta}'_{k},\Lambda_{L}} (i,j) +\mathbf{1}_{i=j=k} 4\varepsilon . 
\]
Expanding in a telescopic sum we get
\[
 \delta_L(E,\varepsilon) =\sum_{k=1}^{|\Lambda_L|}
  \mathbb{E}_{\Lambda_{L}}^{W,\eta^{w}}\left[ \operatorname{tr} \left(
  \rho(\mathcal{H}_{{\beta}'_{k+1},\Lambda_{L}}-E)-\rho(\mathcal{H}_{{\beta}'_k,\Lambda_{L}}-E) \right)\right].
\]
We concentrate now on the \(k\)-th term in the sum.
For a fixed configuration \(\beta_{k^{c}},\) we define
\(y_{k}:= 2\beta_{k}-P_{k,k^{c}}^{W}\mathcal{H}^{-1}_{\beta,\Lambda_{L}\setminus\{k \}}P^{W}_{k^{c},k} .\)
Note that \({\beta}'_{k}={\beta}'_{k} (y_{k},\beta_{k^{c}})\)
is a function of \(y_{k}\) and \(\beta_{k^{c}}.\) 
We consider the function
\(y_{k}\mapsto F_{k} (y_{k}):= \operatorname{tr} \rho(\mathcal{H}_{{\beta}'_{k} (y_{k},\beta_{k^{c}}),\Lambda_{L}}-E).\)
We can write
\begin{align*}
&  \mathbb{E}_{\Lambda_{L}}^{W,\eta^{w}}\left[ \operatorname{tr} \left(
  \rho(\mathcal{H}_{{\beta}'_{k+1},\Lambda_{L}}-E)-\rho(\mathcal{H}_{{\beta}'_k,\Lambda_{L}}-E) \right)\right]\\
&\qquad  = \mathbb{E}_{\Lambda_{L}}^{W,\eta^{w}}\left[
 \int   ( F_k(y_k+4 \varepsilon)-F_k(y_k)) \rho_{a_{k}} (y_{k}) dy_{k}   \right].
\end{align*}
We take the following primitive of \(\rho_{a_{k}}:\)
\[
G_{k}(y):=\int_0^{y}\rho_{a_{k}}(t)dt.
\]
This function is differentiable and satisfies \(G(\infty)=1\) and \(G(0)=0\). Moreover
we can write
\begin{equation*}
\begin{aligned}
&\int_0^{\infty}\rho_{a_{k}}(y)\left[ F_k(y+4 \varepsilon)-F_k(y) \right] d y\\
& =\lim_{M\to \infty}\int_0^M\rho_{a_{k}}(y)\left[ F_k(y+4 \varepsilon)-F_k(y_k) \right] d y=:\lim_{M\to \infty} I_M^{k}.
\end{aligned}
\end{equation*}
Performing  integration by parts, we argue, using \(G_k(0)=0\),
\begin{equation*}
\begin{aligned}
I_M^{k} & = G_{k}(M)(F_k(M+4 \varepsilon)-F_k(M)) -\int_0^M G_{k}(y)(F_k'(y+4 \varepsilon)-F'_k(y))d y\\
&= \int_M^{M+4 \varepsilon} G_{k}(M)F'_k(y)dy - \int_0^M G_{k}(y)F'_k(y+ 4 \varepsilon) d y  +\int_0^M G_{k} (y) F'_k(y)d y.
\end{aligned}
\end{equation*}
We write the second integral as
\[
\int_0^M G_{k}(y)F'_k(y+4 \varepsilon)=\int_{4 \varepsilon}^M G_{k}(y-4 \varepsilon)F'_k(y) d y + \int_M^{M+4 \varepsilon}G_{k}(y- 4 \varepsilon)F'_k(y)dy
\]
and the third integral as
\[
\int_0^M G_{k}(y)F'_k(y)d y = \int_{4 \varepsilon}^M G_k(y)F'_k(y)d y +\int_0^{4 \varepsilon} G_{k}(y)F'_k(y)d y.
\]
Putting everything together we get
\begin{equation*}
\begin{aligned}
I_M^{k} & = \int_M^{M+4 \varepsilon}(G_{k}(M)-G_{k}(y-4 \varepsilon))F'_k(y)d y \\
&+ \int_{4 \varepsilon}^M (G_{k}(y)-G_{k}(y-4 \varepsilon))F'_k(y)d y +\int_0^{4 \varepsilon}G_{k}(y)F'_k(y)d y.
\end{aligned}
\end{equation*}
Now we argue
\[G_{k}(y)-G_{k}(y-4 \varepsilon)=\int_{y-4\varepsilon }^{y} \rho_{a_{k}} (t) dt\leq
 \mathscr{L}_{\rho_{a_{k}}}(4 \varepsilon)\qquad \forall y\in [4\varepsilon,M ].\]
The same bound holds for \(G_{k}(M)-G_{k}(y-4 \varepsilon)\) for \(y\in [M,M+4 \varepsilon]\) and \(G_{k} (y)=G_{k} (y)-G_{k} (0)\) for \(y\in [0,4 \varepsilon]\).
Therefore
\[
I_{M}^{k} \le \mathscr{L}_{\rho_{a_{k}}}(4 \varepsilon) \int_0^{M+4 \varepsilon} F'_k(y)d y
=  \mathscr{L}_{\rho_{a_{k}}}(4 \varepsilon) (F_{k} (M+4\varepsilon )-F_{k} (0)).
\]
Finally,  using a standard argument of rank-one perturbation (see e.g. \cite[Lemma 5.25]{kirsch2007invitation}) \((F_{k} (M+4\varepsilon )-F_{k} (0))\leq 1\) uniformly in \(M\).
The result follows.
\end{proof}

\begin{lemma}\label{le:levy-bound}
It holds, for all \(\varepsilon >0,\) and \(a>0\) 
\begin{equation}\label{eq:levy-bound}
\mathscr{L}_{\rho_{a}}(\varepsilon)\leq \sqrt{\frac{2\varepsilon }{\pi }}.
\end{equation}
Moreover, for \(d\geq 3\)  and \(W\geq W_{0},\) (as defined in Theorem \ref{wegner-type-theorem} )
the following improved estimate holds for all \(\varepsilon >0\) 
\begin{equation}\label{eq:levy-bound-bis}
\mathbb{E}_{\Lambda_{L}}^{W,\eta^{w}}\left[\mathscr{L}_{\rho_{a_{j} (\beta_{j^{c}} )}}(\varepsilon)\right ]\leq \frac{C_{1}}{\sqrt{W}}\varepsilon, \qquad \forall j\in \Lambda_{L},
\end{equation}
where  \(C_{1}>0\) is a constant depending only on the dimension.
\end{lemma}

\begin{proof}
Note that, for all \(y>0\) we have
\[
\rho_{a} (y)={\frac{1}{\sqrt{2\pi}}} e^{-\frac{1}{2 y}(y-a)^2} \frac{1}{\sqrt{y}} \le
\frac{1}{\sqrt{2\pi y}},
\]
and therefore
\begin{align*}
\mathscr{L}_{\rho_{a}}(\varepsilon)&=\sup_{x\geq 0}\rho_{a}([x,x+\varepsilon))\leq 
\sup_{x\geq 0}\int_{x}^{x+\varepsilon } \frac{1}{\sqrt{2\pi y}}\,dy
=\int_{0}^{\varepsilon } \frac{1}{\sqrt{2\pi y}}\,dy= \sqrt{\frac{2\varepsilon }{\pi }}.
\end{align*}
This gives the first bound \eqref{eq:levy-bound}.
To obtain the improved bound \eqref{eq:levy-bound-bis} note that, by \eqref{eq:adefj}
\(a=a_{j} (\beta_{j^{c}})>0\) almost surely, hence
the function \(y\mapsto \rho_{a} (y)\) takes its maximum value
in 
\[
y_{a}:= \frac{1}{2}\left(-1+\sqrt{1+4a^{2}} \right).
\]
Therefore we have \(\mathscr{L}_{\rho_{a}}(\varepsilon)\leq  \rho (y_{a})\varepsilon.\)
Now, using  
\[
\frac{1}{2y_{a}}= \frac{1+\sqrt{1+4a^{2}}}{4a^{2}}\leq
\frac{2+2a}{4a^{2}}=\frac{1}{2}\left(\frac{1}{a^{2}}+\frac{1}{a} \right),
\]
we obtain
\[
 \rho (y_{a})= \frac{1}{\sqrt{2\pi y_{a} }}
 e^{-\frac{1}{2} \frac{(y_{a}-a)^{2}}{y_{a}}}\leq
 \frac{1}{\sqrt{2\pi y_{a} }}\leq \frac{1}{\sqrt{2\pi} }
 \left(\frac{1}{a^{2}}+\frac{1}{a} \right)^{\frac{1}{2}}
 \leq  \frac{1}{\sqrt{2\pi} }
 \left(\frac{1}{a}+\frac{1}{\sqrt{a}} \right).
\]
It follows
\begin{align*}
&\mathbb{E}_{\Lambda_{L}}^{W,\eta^{w}}\left[
\mathscr{L}_{\rho_{a_{j} (\beta_{j^{c}} )}}(\varepsilon) \right]
\leq\ \frac{\varepsilon }{\sqrt{2\pi }} \
\left(\mathbb{E}_{\Lambda_{L}}^{W,\eta^{w}}\left[\frac{1}{a_{j}} \right]
+ \mathbb{E}_{\Lambda_{L}}^{W,\eta^{w}}\left[\frac{1}{\sqrt{a}_{j}}
  \right] \right)\\
&\leq \frac{\varepsilon}{\sqrt{2\pi }}\left( \mathbb{E}_{\Lambda_{L}}^{W,\eta^{w}}\left[\frac{1}{a_{j}} \right]+\mathbb{E}_{\Lambda_{L}}^{W,\eta^{w}}\left[\frac{1}{a_{j}} \right]^{1/2}
 \right).
\end{align*}
The result now follows from Lemma \ref{le:boundH0} below setting
\(C_{1}:=\sqrt{2C_{2}/\pi } \).
\end{proof}

\begin{lemma}\label{le:boundH0}
For \(d\geq 3\) and \(W\geq W_{0},\) (as defined in Theorem \ref{wegner-type-theorem}) it holds
\begin{equation}\label{eq:boundH0}
\mathbb{E}_{\Lambda_{L}}^{W,\eta^{w}}\left[\frac{1}{a_{j} (\beta_{j^{c}})} \right]\leq \frac{C_{2}}{W},
\qquad \forall j\in \Lambda_L,
\end{equation}
where  \(C_{2}>0\) is a constant depending only on the dimension.
\end{lemma}

\begin{proof}
In the case \(j\in \partial \Lambda_{L}\) we have, by  \eqref{eq:adefj},
\(a_{j}\geq \eta^{w}_{j}\geq W\) a.s. and hence \eqref{eq:boundH0} holds with \(C_{2}=1.\)\vspace{0,2cm}

Assume now \(j\in \Lambda_{L}\setminus \partial \Lambda_{L}.\)
Using \eqref{eq:ujasHb1} and \eqref{eq:adefj}  we have
\[
a_{j}= \frac{\sum_{k\in \partial \Lambda_{L}}\mathcal{H}_{\beta,\Lambda_L}^{-1}(j,k)\ \eta^{w}_{k}}{\mathcal{H}_{\beta,\Lambda_L}^{-1}(j,j)}=
\frac{e^{u_{j}}}{\mathcal{H}_{\beta,\Lambda_L}^{-1}(j,j)},
\]
therefore
\[
\mathbb{E}_{\Lambda_{L}}^{W,\eta^{w}}\left[\frac{1}{a_{j}} \right]=
\mathbb{E}_{u,\Lambda_{L}}^{W,\eta^{w}}\left[\mathcal{H}_{\beta (u),\Lambda_{L}}^{-1}(j,j)e^{-u_{j}} \right]=
\frac{1}{W}\mathbb{E}_{u,\Lambda_{L}}^{W,\eta^{w}}\left[D^{-1}(j,j)e^{u_{j}} \right],
\]
where we used \eqref{eq:betau} and \eqref{eq:connection-u-beta}. The matrix \(D=D (u):= e^{u}H_{\beta (u),\Lambda_{L} }e^{u} \) can be  characterized via the quadratic form
\begin{equation}\label{eq:Ddef}
\left< v,D (u) v \right>= \sum_{k\sim k'\in \Lambda_{L}}e^{u_{j}+u_{k}} (\nabla_{kk'}v)^{2}+ \sum_{k\in \Lambda_{L}} \tilde{\eta}^{w}_{k}e^{u_{k}} v_{k}^{2},
\end{equation}
where we defined  \(\tilde{\eta}^{w}_{k}:= \eta^{w}_{k}/W\) and \(\nabla_{kk'}v:= v_{k}-v_{k'}.\)
To estimate the average  of \(D^{-1}(j,j)e^{u_{j}}\) we use the same strategy as
in the proof of Theorem 3 of  \cite{Disertori2010a}).
We can write \(D^{-1}(j,j)e^{u_{j}}= \left< f,D^{-1}f \right>,\)
where \(f_{k}:= \delta_{kj}e^{u_{j}/2}=e^{u_{j}/2} (\delta_{j}) (k).\)
Setting \(D_{0}:=D (0)= -\Delta +\tilde{\eta}^{w}\)  we argue
\begin{align*}
  \left< f,D^{-1}f \right>&= \left< D_{0}D_{0}^{-1}f,D^{-1}f \right>\\
&= \sum_{k\sim k'}  (\nabla_{kk'}D_{0}^{-1}f)  (\nabla_{kk'}D^{-1}f)+
\sum_{k} \tilde{\eta}^{w}_{k} (D_{0}^{-1}f)_{k} (D^{-1}f)_{k}\\
& = \sum_{k\sim k'}  \frac{(\nabla_{kk'}D_{0}^{-1}f)}{e^{(u_{k}+u_{k'})/2}}  \frac{(\nabla_{kk'}D^{-1}f)}{e^{-(u_{k}+u_{k'})/2}}+
\sum_{k}\tilde{\eta}^{w}_{k}  \frac{(D_{0}^{-1}f)_{k}}{e^{u_{k}/2}} \frac{(D^{-1}f)_{k}}{e^{-u_{k}/2}}\\
&\leq \left( \sum_{k\sim k'}  \frac{(\nabla_{kk'}D_{0}^{-1}f)^{2}}{e^{(u_{k}+u_{k'})}}+
  \sum_{k}\tilde{\eta}^{w}_{k}  \frac{(D_{0}^{-1}f)_{k}^{2}}{e^{u_{k}}}  \right)^{\frac{1}{2}}\left< f,D^{-1}f \right>^{\frac{1}{2}},
\end{align*}
where in the last step we used Cauchy-Schwarz inequality. It follows
\begin{align*}
\left<   f,D^{-1}f \right>&\leq  \sum_{k\sim k'}  \frac{(\nabla_{kk'}D_{0}^{-1}f)^{2}}{e^{(u_{k}+u_{k'})}}+
\sum_{k}\tilde{\eta}^{w}_{k}  \frac{(D_{0}^{-1}f)_{k}^{2}}{e^{u_{k}}}\\
&=
 \sum_{k\sim k'}  (\nabla_{kk'}D_{0}^{-1}\delta_{j})^{2}\ e^{u_{j}-(u_{k}+u_{k'})}+
\sum_{k}\tilde{\eta}^{w}_{k}  (D_{0}^{-1}\delta_{j})_{k}^{2}\ e^{u_{j}-u_{k}}
\end{align*}
where we used the explicit form of \(f.\) Therefore
\begin{align*}
  &\mathbb{E}_{u,\Lambda_{L}}^{W,\eta^{w}}\left[\left< f,D^{-1}f \right> \right]\leq
 \sum_{k\sim k'}  (\nabla_{kk'}D_{0}^{-1}\delta_{j})^{2}\ \mathbb{E}_{u,\Lambda_{L}}^{W,\eta^{w}}\left[e^{u_{j}-(u_{k}+u_{k'})} \right]\\
&\qquad +
\sum_{k}\tilde{\eta}^{w}_{k}  (D_{0}^{-1}\delta_{j})_{k}^{2}\   \mathbb{E}_{u,\Lambda_{L}}^{W,\eta^{w}}\left[e^{u_{j}-u_{k}} \right].
\end{align*}
Note that
\begin{align*}
& \mathbb{E}_{u,\Lambda_{L}}^{W,\eta^{w}}\left[e^{u_{j}-(u_{k}+u_{k'})} \right]\leq \ 4\ 
  \mathbb{E}_{u,\Lambda_{L}}^{W,\eta^{w}}\left[(\cosh (u_{j}-u_{k}))^{2}\right]^{\frac{1}{2}}
   \mathbb{E}_{u,\Lambda_{L}}^{W,\eta^{w}}\left[(\cosh u_{k'})^{2}\right]^{\frac{1}{2}}\\
& \mathbb{E}_{u,\Lambda_{L}}^{W,\eta^{w}}\left[e^{u_{j}-u_{k}} \right]\leq \ 2\ 
  \mathbb{E}_{u,\Lambda_{L}}^{W,\eta^{w}}\left[\cosh (u_{j}-u_{k})\right].
\end{align*}
The bounds \eqref{ward2} and \eqref{ward3} in Appendix B
ensure $\mathbb{E}_{u,\Lambda_{L}}^{W,\eta^{w}}\left[(\cosh (u_{j}-u_{k}))^{m} \right]\leq 2$
$\forall j,k\in \Lambda_{L}$
and  \( \mathbb{E}_{u,\Lambda_{L}}^{W,\eta^{w}}\left[(\cosh u_{k})^{2}\right]\leq 8\) $\forall j\in \Lambda_{L}.$
Putting all these bounds together we obtain
\begin{align}\label{eq:final-bound-D}
  &\mathbb{E}_{u,\Lambda_{L}}^{W,\eta^{w}}\left[\left< f,D^{-1}f \right> \right]\leq
16 \left( \sum_{k\sim k'}  (\nabla_{kk'}D_{0}^{-1}\delta_{j})^{2}+\sum_{k}\tilde{\eta}^{w}_{k}
( D_{0}^{-1}\delta_{j})_{k}^{2}  \right)\\
  &= 16 \left< \delta_{j}, D_{0}^{-1}\delta_{j} \right>= 16 (-\Delta_{\Lambda_{L}}+\tilde{\eta})^{-1}_{jj}\leq C_{2},
\nonumber
\end{align}
for some constant \(C_{2}\) independent of \(j\) and \(L, \) since we are in dimension \(d\geq 3.\)
This concludes the proof of the lemma.
\end{proof}

\section{An alternative approach} 
\label{sec:alterapproach}
Some of the above results can also be obtained by using the properties of the \textit{infinite volume} measure  \(\nu^W\), defined in \eqref{eq-laplace-beta-Lambda}. This alternative approach also provides the improved bound  \eqref{eq-NE-nolifhistz-upper-bound-d=3} in Theorem  \ref{corol-upper-bound-on-IDS}. In this section we higlight the main
differences with the finite volume approach
and give the proof of \eqref{eq-NE-nolifhistz-upper-bound-d=3}.
For more details see \cite{theserapenne}.

To explain the strategy we  need to introduce a few preliminary notions and results.
Recall that \(\Lambda_L=[-L,L]^d\cap \mathbb{Z}^d.\)  We define,  for every \(i\in\mathbb{Z}^d\) and  \(L\in \mathbb{N}_{\ge 1},\)
\[
\begin{array}{ll}
\psi_L(i):=1& \text{ if }i\notin\Lambda_L,\\
\psi_L(i):=\sum\limits_{k\in\partial \Lambda_L}\mathcal{H}_{\beta,\Lambda_L}^{-1}(i,k)\eta_{\Lambda_L}^{w}(k) = e^{u_{i} (\beta_{\Lambda_{L}})} 
&\text{ if }i\in\Lambda_L.
\end{array}
\]
The following result is an extract of Theorem 1 in  \cite{Sabot2019}.
\begin{prop}\label{prop:sabotzeng}~\\\vspace{-0.5 cm}
\begin{enumerate}
\item For every \((i,j)\in\mathbb{Z}^d\), \(\left(\mathcal{H}_{\beta,\Lambda_L}^{-1}(i,j)\right)_{L\in\mathbb{N}_{\ge 1}}\) is increasing \(\nu^W\)-a.s. Moreover, it converges toward some almost surely finite random variable which is denoted by \(\hat{G}(i,j)\).
\item For every \(i\in\mathbb{Z}^d\), \((\psi_L(i))_{L\in\mathbb{N}_{\ge 1}}\) is a positive martingale with respect to the filtration \(\left(\sigma(\beta_i, i\in\Lambda_L),L\in\mathbb{N}_{\ge 1}\right)\).
\item For every \(i\in\mathbb{Z}^d\), the bracket of \((\psi_L(i))_{L\in\mathbb{N}_{\ge 1}}\) equals \(\left(\mathcal{H}_{\beta,\Lambda_L}^{-1}(i,i)\right)_{L\in\mathbb{N}_{\ge 1}}\).
In particular , \(\left(\psi_L(i)^2-\mathcal{H}_{\beta,\Lambda_L}^{-1}(i,i)\right)_{L\in\mathbb{N}_{\ge 1}}\) is a martingale  for every \(i\in \mathbb{Z}^d\).
\end{enumerate}

\end{prop}

By Theorem 2 in \cite{Sabot2019}, \(\hat{G}\) is the inverse of the infinite
volume operator  \(\mathcal{H}_{\beta}\) in the following sense:
\(\hat{G} (i,j):=\lim_{\varepsilon\rightarrow 0} (\mathcal{H}_{\beta}+\varepsilon )^{-1}(i,j)\) \(\nu^W\)-a.s.. {Moreover, for every \(i,j\), \(\varepsilon\mapsto(\mathcal{H}_{\beta}+\varepsilon )^{-1}(i,j)\) is increasing \(\nu^W\)-a.s..}
These facts are the key for the construction of the infinite volume environment of the related vertex reinforced jump process. A first application is the improved bound   \eqref{eq-NE-nolifhistz-upper-bound-d=3}.

\begin{prop}[Upper bound on the IDS for large \(W\)]\label{prop:upper-b-IDOSd=3}\hspace{2cm}

For \(d\geq 3\) there exists a \(W_{0}>1\) such that for all \(W\geq W_{0},\)
the function \(E\mapsto N (E)\) satisfies the bound
\[
\begin{aligned}
 N(E,H_{\beta}) \le \ C' E\qquad \forall E>0,
\end{aligned}
\]
for some constant \(C'>0\) independent of \(W.\)
\end{prop}

\begin{proof} Note that \(N (E,H_{\beta })=N (WE,\mathcal{H}_{\beta })=:\widetilde{N} (WE). \)
In the rest of the proof we will work with \(\widetilde{N}.\) 
By Section 3.3 in \cite{AW}, for every bounded continuous function \(f\),
\[\int_{0}^{+\infty}f(u)d\widetilde{N}(u)=\mathbb{E}^W\left[f\left(\mathcal{H}_{\beta}\right)(0,0)\right]\]
where \(f\left(\mathcal{H}_{\beta}\right)\) is an operator which is well-defined because \(\mathcal{H}_{\beta}\) is self adjoint.
In particular, for every \(\varepsilon>0\), it holds that
\begin{align}
  \int_0^{+\infty}\frac{1}{u+\varepsilon}d\widetilde{N}(u)&=\mathbb{E}^W\left[(\mathcal{H}_{\beta} +\varepsilon)^{-1}(0,0)\right].\label{int2bis}
\end{align}
Furthermore, as remarked above, 
\((\mathcal{H}_{\beta} +\varepsilon)^{-1}(0,0)\underset{\varepsilon\rightarrow 0}
\longrightarrow\hat{G}(0,0)\), \(\nu^W\)-a.s. and this convergence is increasing.
Therefore, by monotone convergence theorem, for every \(\varepsilon>0\),
\[\mathbb{E}^W\left[(\mathcal{H}_{\beta} +\varepsilon)^{-1}(0,0)\right]\underset{\varepsilon\rightarrow 0}
\longrightarrow \mathbb{E}^W\left[\hat{G}(0,0)\right] \]
and
\[
\int_0^{+\infty}\frac{1}{u+\varepsilon}d\widetilde{N}(u)\underset{\varepsilon\rightarrow 0}\longrightarrow\int_0^{+\infty}\frac{1}{u}d\widetilde{N}(u).
\]
Thus, if we make \(\varepsilon\) go to \(0\) in \eqref{int2bis}, this implies that, \(\nu^W\)-a.s,
\[
\int_0^{+\infty}\frac{1}{u}d\widetilde{N}(u)=\mathbb{E}^W\left[\hat{G}(0,0)\right].
\]
Using Fatou's lemma, it yields
\begin{align*}
&\int_0^{+\infty}\frac{1}{u}d\widetilde{N}(u)=\mathbb{E}^W\left[\hat{G}(0,0)\right]
\leq \underset{L\rightarrow+\infty}\liminf\hspace{0.1 cm}\mathbb{E}^W\left[ \mathcal{H}_{\beta,\Lambda_L}^{-1}(0,0)\right]\\
&\quad =\underset{L\rightarrow+\infty}\liminf\hspace{0.1 cm}\frac{1}{W}\mathbb{E}_{u,\Lambda_{L}}^{W,\eta^{w}}\left[D^{-1}(0,0)e^{2u_{0}} \right]
  =\underset{L\rightarrow+\infty}\liminf\hspace{0.1 cm}\frac{1}{W} \mathbb{E}_{u,\Lambda_{L}}^{W,\eta^{w}}\left[\left< f,D^{-1}f \right> \right],
\end{align*}
where the matrix \(D\) was defined in \eqref{eq:Ddef} and \(f_{k}:= \delta_{k0}e^{u_{0}}\)
(instead of \(f_{k}:= \delta_{kj}e^{u_{j}/2}\)  in the proof of Lemma \ref{le:boundH0}).
Repeating the same arguments as in the proof of Lemma \ref{le:boundH0} we obtain
\begin{align*}
&\mathbb{E}_{u,\Lambda_{L}}^{W,\eta^{w}}\left[\left<f,D^{-1}f\right> \right]\leq
 \sum_{k\sim k'}  (\nabla_{kk'}D_{0}^{-1}\delta_{0})^{2}\ \mathbb{E}_{u,\Lambda_{L}}^{W,\eta^{w}}
 \left[e^{2u_{0}-(u_{k}+u_{k'})} \right]\\
&\qquad +
\sum_{k}\tilde{\eta}^{w}_{k}  (D_{0}^{-1}\delta_{0})_{k}^{2}\
\mathbb{E}_{u,\Lambda_{L}}^{W,\eta^{w}}\left[e^{2u_{0}-u_{k}} \right],
\end{align*}
where remember that \(D_{0}:=D (0)=-\Delta_{\Lambda_{L}}+\tilde{\eta}^w\) and \(\tilde{\eta}^{w}_{k}:= \eta^{w}_{k}/W\). Note that
\begin{align*}
& \mathbb{E}_{u,\Lambda_{L}}^{W,\eta^{w}}\left[e^{2u_{0}-(u_{k}+u_{k'})} \right]\leq \ 4\ 
  \mathbb{E}_{u,\Lambda_{L}}^{W,\eta^{w}}\left[(\cosh (u_{0}-u_{k}))^{2}\right]^{\frac{1}{2}}
   \mathbb{E}_{u,\Lambda_{L}}^{W,\eta^{w}}\left[(\cosh (u_{0}-u_{k'}))^{2}\right]^{\frac{1}{2}}\\
& \mathbb{E}_{u,\Lambda_{L}}^{W,\eta^{w}}\left[e^{2u_{0}-u_{k}} \right]\leq \ 4\ 
\mathbb{E}_{u,\Lambda_{L}}^{W,\eta^{w}}\left[(\cosh (u_{0}-u_{k}))^{2}\right]^{\frac{1}{2}}
\mathbb{E}_{u,\Lambda_{L}}^{W,\eta^{w}}\left[(\cosh u_{0})^{2}\right]^{\frac{1}{2}}.
\end{align*}
Using \eqref{ward2} and Lemma \ref{le-ward} we obtain (cf. \eqref{eq:final-bound-D}) 
\[
\mathbb{E}_{u,\Lambda_{L}}^{W,\eta^{w}}\left[\left<f,D^{-1}f\right> \right]\leq
16 \left< \delta_{0}, D_{0}^{-1}\delta_{0} \right>= 16 (-\Delta_{\Lambda_{L}}+\tilde{\eta}^w)^{-1}(0,0)\leq C_{2},
\]
where  \(C_{2}>0\) is the same constant we obtained in  \eqref{eq:final-bound-D} and we used that  we are in dimension \(d\geq 3.\) 
  Hence \(\int_0^{+\infty}\frac{1}{u}d\widetilde{N}(u)\leq C_{2}/W.\) It follows
\begin{align*}
N (E,H_{\beta })=\widetilde{N} (WE)&=\int_0^{WE}\frac{u}{u}d\widetilde{N}(u) \leq WE\int_0^{+\infty}\frac{1}{u}d\widetilde{N}(u)\ \leq\ C_{2} E.
\end{align*}
This concludes the proof setting \(C':=C_{2}\).
\end{proof}

\paragraph{Remark.} Note that  \(\left(\mathcal{H}_{\beta,\Lambda_L}^{-1}(0,0)\right)_{L\in\mathbb{N}_{\ge 1}}\) is the quadratic variation of
the martingale \((\psi_L(0))_{L\in\mathbb{N}_{\ge 1}}\) (cf  Proposition \ref{prop:sabotzeng} ). This observation gives the slightly weaker estimate
\begin{align*}
\int_0^{+\infty}\frac{1}{u}d\widetilde{N}(u)&=\mathbb{E}^W\left[\hat{G}(0,0)\right]
\leq \underset{L\rightarrow+\infty}\liminf\hspace{0.1 cm}\mathbb{E}^W\left[ \mathcal{H}_{\beta,\Lambda_L}^{-1}(0,0)\right]\\
&=\underset{L\rightarrow+\infty}\liminf\hspace{0.1 cm}\mathbb{E}^W\left[\psi_L(0)^2\right]\leq 16,
\end{align*}
where in the last inequality we used Lemma \ref{le-ward} together with
\[
\mathbb{E}^W\left[\psi_L(0)^2\right]=  \mathbb{E}_{u,\Lambda_{L}}^{W,\eta^{w}}\left[e^{2u_{0}}\right]\leq 4
\mathbb{E}_{u,\Lambda_{L}}^{W,\eta^{w}}\left[(\cosh u_{0})^{2}\right]\leq 16.
\]


The infinite volume measure approach also gives an alternative proof of the lower bound for \(N (E).\) For this we need some more definitions. 
 Setting for \(i\in\mathbb{Z}^d\) we define
\begin{equation}\label{beta-tilde}
\tilde{\beta}_i:=\beta_i-\delta_{i,0}\frac{1}{2\hat{G}(0,0)}
\end{equation}
we have the following result.

\begin{prop}[Proposition 2.4 in \cite{Gerard}]\label{fact:ger}
Recall the definition of \(W_{cr}\) in \eqref{def:Wcrit}. Then, 
for all \( W<W_{cr},\) 
\(\frac{1}{2\hat{G}(0,0)}\) has density \(\mathbbm{1}_{\gamma > 0} \frac{1}{\sqrt{\pi \gamma}}e^{-\gamma}\).
Moreover, \(\tilde{\beta}\) and \(\frac{1}{2\hat{G}(0,0)}\) are independent random variables.
\end{prop}
Note that this proposition works for any \(W\) such that the corresponding reinforced jump process is recurrent. This is true in particular for \(W<W_{cr}.\)
The variable  \(\tilde{\beta}_{i}\)  arises naturally as the jump rate of the vertex reinforced jump process at vertex \(i\)
(See \cite[Theorem 1.(iii)]{Sabot2019}).
In the following we will consider  \(\tilde{\mathcal{H}}_{\beta}:=2\tilde{\beta}-P^{W}\) and its  Dirichlet restriction on the finite box \(\Lambda_L\)
\(\tilde{\mathcal{H}}_{\beta,\Lambda_L}^D\) (cf  \eqref{def:HD}).

Finally, recall that the graph \(\Lambda_{L}\cup \delta \) has vertex set \(\Lambda_{L} \cup \{\delta  \}\)
and edge set \( E (\Lambda_{L} )\cup \{\{i,\delta  \}|\ i\in \Lambda_{L}  \},\) and we
defined \(W_{i,\delta }=\eta_{i}^{w}=\sum_{j\sim i, j\notin\Lambda_L} W\) \(\forall i\in \Lambda_{L} .\)
Now consider an electrical network on \(\Lambda_{L}\cup \delta\) with conductances 
\begin{align*}
c (i,j)&:=W\frac{\hat{G}(0,i)\hat{G}(0,j)}{\hat{G}(0,0)^2}\qquad\qquad  \forall i,j\in \Lambda_L,\\
c(i,\delta_L)&:=\sum\limits_{\substack{j\sim i\\j\notin\Lambda_L}} W\left(\frac{\hat{G}(0,i)\hat{G}(0,j)}{\hat{G}(0,0)^2}+\frac{\hat{G}(0,i)^2}{\hat{G}(0,0)^2}\right)
\qquad \forall i\in \Lambda_L,
\end{align*}
and let \(\mathcal{R}_L(0\longleftrightarrow\delta)\) be the effective resistance of the random walk associated with this network.
The following proposition is proved in \cite{theserapenne}.

\begin{prop}\label{prop:resistance}
Let \(W<W_{cr}.\) Then, for every \(L\in\mathbb{N}_{\ge 1}\),
\[
(\tilde{\mathcal{H}}_{\beta,\Lambda_L}^D)^{-1}(0,0)=\mathcal{R}_L(0\longleftrightarrow\delta).
\]
\end{prop}
Thanks to this result, we can use some tricks from the theory of electrical networks (e.g. \cite[Chap. 2]{LP}) to construct an alternative proof
of  Theorem \ref{thm-main-theorem}. We sketch below the argument.


\begin{proof}[Proof of Theorem \ref{thm-main-theorem} (II)]
As in the proof given in section \ref{sec:org8c54fb3}, we start from
\[
N(E,\mathcal{H}_{\beta})\geq \frac{1}{|\Lambda_L|}\nu^W\left((\mathcal{H}_{\beta,\Lambda_L}^D)^{-1}(0,0)\geq \frac{1}{E} \right).
\]
Note that
\[
\begin{array}{l}
 \nu^W\left(\frac{1}{2\hat{G}(0,0)}\leq \frac{E}{4} \right)=\nu^W\left(\hat{G}(0,0)\geq\frac{2}{E}\right)\\
 \leq \nu^W\left(\hat{G}(0,0)-(\mathcal{H}_{\beta,\Lambda_L}^D)^{-1}(0,0)\geq \frac{1}{E}\right)+\nu^W\left( (\mathcal{H}_{\beta,\Lambda_L}^D)^{-1}(0,0)\geq \frac{1}{E}\right).
\end{array}
\]
Consequently,
\begin{align*}
N(E,\mathcal{H}_{\beta})&\geq \frac{1}{|\Lambda_L|}\left[\nu^W\left(\tfrac{1}{2\hat{G}(0,0)}\leq \tfrac{E}{4} \right)-
\nu^W\left(\hat{G}(0,0)-(\mathcal{H}_{\beta,\Lambda_L}^D)^{-1}(0,0)\geq \tfrac{1}{E}\right)\right]\\
&\geq \frac{1}{|\Lambda_L|}\left[C \sqrt{E}-
\nu^W\left(\hat{G}(0,0)-(\mathcal{H}_{\beta,\Lambda_L}^D)^{-1}(0,0)\geq \tfrac{1}{E}\right)\right],
\end{align*}
where \(C>0\) is some constant and we used  that \(\frac{1}{2\hat{G}(0,0)}\) is a \(\Gamma(1/2,1)\) random variable (cf. Prop. \ref{fact:ger}).
We claim that, for \(E\) small and \(L\) large enough,
\[
\nu^W\left(\hat{G}(0,0)-(\mathcal{H}_{\beta,\Lambda_L}^D)^{-1}(0,0)\geq \frac{1}{E}\right)\ll  \sqrt{E},
\]
which implies the result. To prove this asymptotic domination, note that 
 \(\tilde{\mathcal{H}}_{\beta,\Lambda_L}^D-\mathcal{H}_{\beta,\Lambda_L}^D=-\delta_{0} \frac{1}{\hat{G}(0,0)}\) 
 (cf. equation \eqref{beta-tilde}). Therefore 
\begin{equation*}
(\mathcal{H}_{\beta,\Lambda_L}^D)^{-1}(0,0)=
\frac{(\tilde{\mathcal{H}}_{\beta,\Lambda_L}^D)^{-1}(0,0)}{1+(\tilde{\mathcal{H}}_{\beta,\Lambda_L}^D)^{-1}(0,0)/\hat{G}(0,0)} 
\end{equation*}
and thus, using  Proposition \ref{prop:resistance},
\begin{align*}
\hat{G}(0,0)-(\mathcal{H}_{\beta,\Lambda_L}^D)^{-1}(0,0)&\leq \frac{\hat{G}(0,0)^2}{(\tilde{\mathcal{H}}_{\beta,\Lambda_L}^D)^{-1}(0,0)}=
 \frac{\hat{G}(0,0)^2}{\mathcal{R}_L(0\longleftrightarrow\delta)}.
\end{align*}
Therefore, we have
\begin{align}
\nu^W\left(\hat{G}(0,0)-(\mathcal{H}_{\beta,\Lambda_L}^D)^{-1}(0,0)\geq \frac{1}{E}\right)&\leq \nu^W\left( \frac{\hat{G}(0,0)^2}{\mathcal{R}_L(0\longleftrightarrow\delta)}\geq \frac{1}{E}\right)\nonumber\\
&\hspace{-3 cm}=\int_0^{+\infty}\frac{e^{-\gamma}}{\sqrt{\pi\gamma}}\nu^W\left(\frac{1}{\mathcal{R}_L(0\longleftrightarrow\delta)}\geq\frac{4\gamma^2}{E} \right)d\gamma\label{eq:integral-bound}
\end{align}
where in the last equality of \eqref{eq:integral-bound}, we used the Proposition \ref{fact:ger} and the measurability of \((\mathcal{H}_{\beta,\Lambda_L}^D)^{-1}(0,0)=\mathcal{R}_L(0\longleftrightarrow\delta)\) with respect to \(\tilde{\beta}\).
Then, one can use classical results in electrical networks (The Nash-Williams inequality) to control the inverse of \(\mathcal{R}_L(0\longleftrightarrow\delta)\), using the local conductances on the boundary of \(\Lambda_L\). Finally, we can control these local conductances thanks to Corollary \ref{cor:fractional} if we choose a "good" \(L\) as a function of \(E\).
\end{proof}

\appendix

\section{Monotonicity}

The following monotonicity result can be found  in \cite[Theorem 6]{Poudevigne2019}

\begin{theorem}\label{th:monotonicity}
Let \(V\) be a finite set, \(W^{+},W^{-}\in \mathbb{R}_{\geq 0}^{V\times V}\)
two families of non-negative weights satisfying
\(W^{\pm}_{jj}=0\) \(\forall j\in V\) and 
\[
W^{-}_{ji}=W^{-}_{ij}\leq W^{+}_{ij}=W^{+}_{ji}\qquad \forall i\neq j.
\]
Let \(E^{+}\) (resp. \(E^{-}\)) be the set of pairs with positive weight \(W^{+}_{ij}>0\) (resp  \(W^{-}_{ij}>0\)) and denote by \(\mathcal{G}^{\pm}= (V,E^{\pm})\) the corresponding graphs.

If \(i,j\in V\) are connected by \(\mathcal{G}^{-},\) then it holds
\begin{equation}\label{eq:monotonicity}
\mathbb{E}_{\mathcal{G}^{-}}^{W^{-},0}\left[ f\left(
\frac{\mathcal{H}_{\beta,V,W^{-}}^{-1}(j,k)}{\mathcal{H}_{\beta,V,W^{-}}^{-1}(j,j)} \right)\right]
\le \mathbb{E}_{\mathcal{G}^{+}}^{W^{+},0}\left[f\left(\frac{\mathcal{H}_{\beta,V,W^{+}}^{-1}(j,k)}{\mathcal{H}_{\beta,V,W^{+}}^{-1}(j,j)} \right)\right]
\end{equation}
for any concave function \(f.\) Here we write \(\mathcal{H}_{\beta,V,W^{\pm}}\) instead of \(\mathcal{H}_{\beta,V}\) to emphasize the dependence of \(W^{\pm}.\)
\end{theorem}

In this paper we use the following corollary.

\begin{corollary}
\label{corollary-monotonicity}
Let \(\mathcal{G}=(V,E)\) be a connected finite graph and \(W\in  \mathbb{R}_{> 0}^E\) a given set of weights. Fix a vertex \(j_{0}\in V\) and set \(\eta_{j}=\eta \delta_{jj_{0}}\) with \(\eta>0\) (one pinning at \(j_{0}\)). It holds, for all \(j\in V,\)
\[
\mathbb{E}_{\mathcal{G}}^{W,\eta}\left[ \sqrt{\frac{\mathcal{H}_{\beta,V}^{-1}(j_{0},j)}{\mathcal{H}_{\beta,V}^{-1}(j_{0},j_{0})}} \right]\geq \mathbb{E}_{\mathcal{G}}^{W,0}\left[ \sqrt{\frac{\mathcal{H}_{\beta,V}^{-1}(j_{0},j)}{\mathcal{H}_{\beta,V}^{-1}(j_{0},j_{0})}} \right].
\]

\end{corollary}

\begin{proof}
The measure \(\nu_{\mathcal{G}}^{W,\eta} \) is the marginal of \(\nu_{\mathcal{G}^{\delta }}^{W,0} ,\)
where the graph 
\(\mathcal{G}^{\delta }\) has vertex set \(V  \cup \{\delta  \}\) and edge set
\(E \cup \{j_{0},\delta \} ,\) and we defined \(W_{j_{0},\delta }=\eta.\)
Moreover,  by resolvent expansion, we have
\begin{align*}
\mathcal{H}_{\beta,V\cup \{\delta  \}}^{-1}(j_{0},j)&=
\mathcal{H}_{\beta,V}^{-1}(j_{0},j) \ \left[ 1+\eta^{2} \mathcal{H}_{\beta,V}^{-1}(j_{0},j_{0})
\mathcal{H}_{\beta,V\cup \{\delta  \}}^{-1}(\delta ,\delta )
\right],\\
\mathcal{H}_{\beta,V  \cup \{\delta  \}}^{-1}(j_{0},j_{0})&=
\mathcal{H}_{\beta,V}^{-1}(j_{0},j_{0}) \ \left[ 1+\eta^{2} \mathcal{H}_{\beta,V}^{-1}(j_{0},j_{0})
\mathcal{H}_{\beta,V\cup \{\delta  \}}^{-1}(\delta ,\delta ),
\right]
\end{align*}
and hence
\[
\mathbb{E}_{\mathcal{G}}^{W,\eta}\left[ \sqrt{\tfrac{\mathcal{H}_{\beta,V}^{-1}(j_{0},j)}{\mathcal{H}_{\beta,V}^{-1}(j_{0},j_{0})}} \right]=
\mathbb{E}_{\mathcal{G}^{\delta }}^{W,0}\left[ \sqrt{\tfrac{\mathcal{H}_{\beta,V}^{-1}(j_{0},j)}{\mathcal{H}_{\beta,V}^{-1}(j_{0},j_{0})}} \right]=
\mathbb{E}_{\mathcal{G}^{\delta }}^{W,0}\left[ \sqrt{\tfrac{\mathcal{H}_{\beta,V\cup \{\delta  \}}^{-1}(j_{0},j)}{\mathcal{H}_{\beta,V\cup \{\delta  \}}^{-1}(j_{0},j_{0})}} \right]
\]
Define \(\widetilde{W}_{ij}=W_{ij}\) \(\forall i\sim j\in V\) and \(\widetilde{W}_{j_{0}\delta }=0.\)
Then \(W_{ij}\geq \widetilde{W}_{ij}\) \(\forall i\sim j\in \mathcal{G}^{\delta }\) and the graph
generated by \(\widetilde{W}\) connects \(j_{0}\) to \(j\) \(\forall j\in V.\)
Since \(f (x)=\sqrt{x}\) is a concave function, by Theorem \ref{th:monotonicity} we have
\[
\mathbb{E}_{\mathcal{G}^{\delta }}^{W,0}\left[ \sqrt{\tfrac{\mathcal{H}_{\beta,V\cup \{\delta  \},W}^{-1}(j_{0},j)}{\mathcal{H}_{\beta,V\cup \{\delta  \},W}^{-1}(j_{0},j_{0})}} \right]\geq 
\mathbb{E}_{\mathcal{G}^{\delta }}^{\widetilde{W},0}\left[ \sqrt{\tfrac{\mathcal{H}_{\beta,V\cup \{\delta  \},\widetilde{W}}^{-1}(j_{0},j)}{\mathcal{H}_{\beta,V\cup \{\delta  \},\widetilde{W}}^{-1}(j_{0},j_{0})}} \right].
\]
Since \(\widetilde{W}_{j_{0}\delta }=0\) and \(\widetilde{W}=W\) on \(V,\) \(2\beta_{\delta }\) is independent of the other random variables and we have
\(\mathcal{H}_{\beta,V\cup \{\delta  \},\widetilde{W}}=2\beta_{\delta }\oplus \mathcal{H}_{\beta,V},\)
where we abbreviated \(\mathcal{H}_{\beta,V}=\mathcal{H}_{\beta,V,W}.\)
Therefore
\[
\mathcal{H}_{\beta,V\cup \{\delta  \},\widetilde{W}}^{-1}(j_{0},j)= \mathcal{H}_{\beta,V}^{-1}(j_{0},j),\qquad 
\mathcal{H}_{\beta,V\cup \{\delta  \},\widetilde{W}}^{-1}(j_{0},j_{0})= \mathcal{H}_{\beta,V}^{-1}(j_{0},j_{0}).
\]
It follows
\[
\mathbb{E}_{\mathcal{G}^{\delta }}^{\widetilde{W},0}\left[ \sqrt{\tfrac{\mathcal{H}_{\beta,V\cup \{\delta  \},\widetilde{W}}^{-1}(j_{0},j)}{\mathcal{H}_{\beta,V\cup \{\delta  \},\widetilde{W}}^{-1}(j_{0},j_{0})}} \right]=
\mathbb{E}_{\mathcal{G}^{\delta }}^{\widetilde{W},0}\left[ \sqrt{\tfrac{\mathcal{H}_{\beta,V}^{-1}(j_{0},j)}{\mathcal{H}_{\beta,V}^{-1}(j_{0},j_{0})}} \right]=
\mathbb{E}_{\mathcal{G}}^{W,0}\left[ \sqrt{\tfrac{\mathcal{H}_{\beta,V}^{-1}(j_{0},j)}{\mathcal{H}_{\beta,V}^{-1}(j_{0},j_{0})}} \right],
\]
where in the last step we used that \(\beta_{\delta }\) is independent of the other variables.
This concludes the proof.
\end{proof}


\section{Long range order estimates on the $u$ field associated to the \( H^{2|2} \)-model}
Recall the definition of $\mu^{W, \eta}_{\Lambda} (u )$  and  $\eta^{w}$ in \eqref{def:umeasure} and \ref{eq-eta-boundary-Lambda}
respectively.

\begin{lemma}\label{le-ward1}
For  any \(W>0,\) $d\geq 1,$ \( j\in  \Lambda\) such that $\eta_{j}>0$ and \(m\leq \eta_{j}/2\) we have
\begin{equation}\label{ward0}
\mathbb{E}_{u,\Lambda}^{W,\eta}\left[(\cosh u_{j})^{m} \right]\leq
\frac{1}{1-\frac{m}{\eta_{j}}}\leq 2.
\end{equation}
In particular, in the case $\Lambda=\Lambda_{L}$ and $\eta=\eta^{w}$ we have 
\begin{equation}\label{ward1}
\mathbb{E}_{u,\Lambda_{L}}^{W,\eta^{w}}\left[(\cosh u_{j})^{m} \right]\leq
\frac{1}{1-\frac{m}{W}}\leq 2
\end{equation}
for all $j\in  \partial \Lambda_{L}$ and $m\leq W/2.$
\end{lemma}

\begin{proof}
This inequality follows by a supersymmetric Ward identity analog to the one in section 5.1 of
\cite{Disertori2010a}.  See also Lemma 4.2 in \cite{disertori-merkl-rolles-2024transiencevertexreinforcedjumpprocesses}.
We sketch here the argument. We refer to the notation in \cite{Disertori2010a},
in particular
\[
B_j=\cosh u_j + \tfrac{1}{2}s_j^2 e^{u_j}
\] where \(s_j\) is a real variable. Denoting the supersymmetric mean by \(\left< \cdot \right>_{\text{susy}}\), we have
\[
1=\left< (B_{j}+\bar{\psi}_{j} \psi_{j}e^{u_{j}})^{m}\right>_{\text{susy}}=
\mathbb{E}_{u,\Lambda}^{W,\eta}\left[B_{j}^{m}
\Big (1-\frac{m}{B_{j}}\mathcal{H}_{\beta,\Lambda_{L}}^{-1}(j,j) \Big ) \right]
\]
where \(\overline{\psi}_j,\psi_j\) are anti-commuting variables.
The bound \eqref{ward0} now follows from \(B_{j}\geq 1\) and 
\(\mathcal{H}_{\beta,\Lambda_L}^{-1}(j,j)\leq \frac{1}{\eta_{j}}.\)
The bound \eqref{ward1} is obtained observing that $\eta^{w}_{j}\geq W$ 
\(\forall j\in  \partial \Lambda_{L}.\)

\end{proof}

The following result has been proved in  \cite[Theorem 1]{Disertori2010a}.

\begin{theorem}\label{th:bound-u-fluctDSZ}
Let \(d\geq 3.\)
There exists \(W_{0}'=W_{0}' (d) >1\) such that \(\forall W\geq W_{0}'\) we have
\begin{align}\label{ward2}
\mathbb{E}_{u,\Lambda_{L}}^{W,\eta}\left[(\cosh (u_{j}-u_{k}))^{m} \right]\leq 2\qquad
\forall j,k\in  \Lambda_{L} ,\  m\leq W^{\frac{1}{8}}.
\end{align}
This bound holds for all $\eta.$
\end{theorem}
\begin{proof}
Although the model considered in  \cite{Disertori2010a} has uniform pinning
\(\eta_{j}=\varepsilon >0\) \(\forall j\in \Lambda_{L}\) the proof of Theorem 1 is completely
independent from the pinning choice.
Also, while in \cite{Disertori2010a} only \(d=3\) is considered,
the same proof works for any \(d\geq 3.\) Indeed the key dimension-dependent result (Lemma 5) is proved
for general  dimension \(d\geq 3.\) The same strategy was used in \cite{Disertori2015tr}
in the case when the edge weights $W_{e}$ are independent Gamma distributed variables. 
For a related result on a one dimensional chain with non homogeneous weights, see
\cite{DMR2024fluctuationsnonlinearsupersymmetrichyperbolic}.
\end{proof}


\begin{lemma}\label{le-ward}
Let \(d\geq 3\) and  \(W_{0}'=W_{0}' (d) >1\) be the parameter introduced in Theorem \ref{th:bound-u-fluctDSZ} above.
For all \(W\geq W_{0}:=\max\{ W_{0}', 4^{8}\}\) we have
\begin{equation}\label{ward3}
 \mathbb{E}_{u,\Lambda_{L}}^{W,\eta^{w}}\left[(\cosh u_{k})^{2}\right]\leq 8\qquad \forall k\in \Lambda_{L}.
\end{equation}
\end{lemma}

\begin{proof}
By Lemma \ref{le-ward1}, for  any \(W\geq W_{0},\) \( j\in \partial \Lambda\) and \(m\leq W/2\) we have
\begin{equation}\label{ward1bis}
\mathbb{E}_{u,\Lambda_{L}}^{W,\eta^{w}}\left[(\cosh u_{j})^{m} \right]\leq
\frac{1}{1-\frac{m}{W}}\leq 2.
\end{equation}
Fix now  \(k\in \Lambda_{L} \setminus \partial \Lambda_{L} \) and let \(j\) be some vertex on \( \partial \Lambda_{L}.\)
We have the bound \(\cosh u_{k}\leq 2 \cosh u_{j} \cosh (u_{k}-u_{j})\) and hence
\[
\mathbb{E}_{u,\Lambda_{L}}^{W,\eta^{w}}\left[(\cosh u_{k})^{2} \right]\leq
\ 4 \ \mathbb{E}_{u,\Lambda_{L}}^{W,\eta^{w}}\left[(\cosh u_{j})^{4} \right]^{\frac{1}{2}}\mathbb{E}_{u,\Lambda_{L}}^{W,\eta^{w}}\left[(\cosh (u_{j}-u_{k}))^{4} \right]^{\frac{1}{2}}.
\]
The constraint $ W\geq \max\{ W_{0}, 4^{8}\}$ ensures $4\leq W^{\frac{1}{8}}$ and $4\leq \frac{W}{2}.$
The result now follows from \eqref{ward2} and  \eqref{ward1bis}.

\end{proof}

\textbf{Remark} Note that in  \cite{Disertori2010a}
the bound \eqref{ward3} is proved in Theorem 2 and requires quite some work due to the presence of
uniform \textit{small} pinning \(\varepsilon\sim 1/|\Lambda_{L}|^{1-s},\) \(0<s\ll 1.\) 
Here the same bound follows easily from \eqref{ward2} and the fact that we have
\textit{large} pinning at the boundary \(\eta^{w}_{j}\geq W\gg 1\) \(\forall j\in \partial \Lambda_{L}.\)

\bibliographystyle{plain}


\begin{thebibliography}{10}

\bibitem{AW}
Michael Aizenman and Simone Warzel.
\newblock {\em Random operators}, volume 168 of {\em Graduate Studies in
  Mathematics}.
\newblock American Mathematical Society, Providence, RI, 2015.
\newblock Disorder effects on quantum spectra and dynamics.

\bibitem{AK23}
B.~L. Altshuler and V.~E. Kravtsov.
\newblock Random cantor sets and mini-bands in local spectrum of quantum
  systems.
\newblock {\em arXiv:2301.12279 [cond-mat]}, 2023.

\bibitem{Anderson58}
P.~W. Anderson.
\newblock Absence of diffusion in certain random lattices.
\newblock {\em Phys. Rev.}, 109:1492--1505, 1958.

\bibitem{Angel2014lo}
Omer Angel, Nicholas Crawford, and Gady Kozma.
\newblock Localization for linearly edge reinforced random walks.
\newblock {\em Duke Math. J.}, 163(5):889--921, 2014.

\bibitem{Anves2009}
Ton\'{c}i Antunovi\'{c} and Ivan Veseli\'{c}.
\newblock Spectral asymptotics of percolation {H}amiltonians on amenable
  {C}ayley graphs.
\newblock In {\em Methods of spectral analysis in mathematical physics}, volume
  186 of {\em Oper. Theory Adv. Appl.}, pages 1--29. Birkh\"{a}user Verlag,
  Basel, 2009.

\bibitem{bakerlossstolz09}
Jeff Baker, Michael Loss, and G\"{u}nter Stolz.
\newblock Low energy properties of the random displacement model.
\newblock {\em J. Funct. Anal.}, 256(8):2725--2740, 2009.

\bibitem{Collevecchio2018}
Andrea Collevecchio and Xiaolin Zeng.
\newblock A note on recurrence of the vertex reinforced jump process and
  fractional moments localization.
\newblock {\em Electron. J. Probab.}, 26:Paper No. 63, 16, 2021.

\bibitem{DMR2024fluctuationsnonlinearsupersymmetrichyperbolic}
Margherita Disertori, Franz Merkl, and Silke W.~W. Rolles.
\newblock Fluctuations in the non-linear supersymmetric hyperbolic sigma model
  with long-range interactions.
\newblock {\em arXiv:2408.08136 [math-phys]}, 2024.

\bibitem{disertori-merkl-rolles-2024transiencevertexreinforcedjumpprocesses}
Margherita Disertori, Franz Merkl, and Silke W.~W. Rolles.
\newblock Transience of vertex-reinforced jump processes with long-range jumps.
\newblock {\em arXiv:2305.07359 [math]}, 2024.

\bibitem{Disertori2015tr}
Margherita Disertori, Christophe Sabot, and Pierre Tarr\`es.
\newblock Transience of edge-reinforced random walk.
\newblock {\em Comm. Math. Phys.}, 339(1):121--148, 2015.

\bibitem{Disertori2010}
Margherita Disertori and Thomas Spencer.
\newblock {A}nderson localization for a supersymmetric sigma model.
\newblock {\em Comm. Math. Phys.}, 300(3):659--671, 2010.

\bibitem{Disertori2010a}
Margherita Disertori, Thomas Spencer, and Martin~R. Zirnbauer.
\newblock Quasi-diffusion in a 3{D} supersymmetric hyperbolic sigma model.
\newblock {\em Comm. Math. Phys.}, 300(2):435--486, 2010.

\bibitem{drunk1992migdal}
W.~Drunk, D.~Fuchs, and Martin~R. Zirnbauer.
\newblock Migdal-{K}adanoff renormalization of a nonlinear supervector model
  with hyperbolic symmetry.
\newblock {\em Ann. Physik (8)}, 1(2):134--150, 1992.

\bibitem{elgart2014ground}
Alexander Elgart and Abel Klein.
\newblock Ground state energy of trimmed discrete {S}chr\"{o}dinger operators
  and localization for trimmed {A}nderson models.
\newblock {\em J. Spectr. Theory}, 4(2):391--413, 2014.

\bibitem{ErdoesSchoeder2014}
L\'{a}szl\'{o} Erd\H{o}s and Dominik Schr\"{o}der.
\newblock Phase transition in the density of states of quantum spin glasses.
\newblock {\em Math. Phys. Anal. Geom.}, 17(3-4):441--464, 2014.

\bibitem{Fukushima-Ueki2010}
Ryoki Fukushima and Naomasa Ueki.
\newblock Classical and quantum behavior of the integrated density of states
  for a randomly perturbed lattice.
\newblock {\em Ann. Henri Poincar\'{e}}, 11(6):1053--1083, 2010.

\bibitem{Gerard}
Thomas Gerard.
\newblock Representations of the vertex reinforced jump process as a mixture of
  {Markov} processes on {{\(\mathbb{Z}^d\)}} and infinite trees.
\newblock {\em Electron. J. Probab.}, 25:45, 2020.
\newblock Id/No 108.

\bibitem{kirsch2007invitation}
Werner Kirsch.
\newblock An invitation to random {S}chr\"{o}dinger operators.
\newblock In {\em Random {S}chr\"{o}dinger operators}, volume~25 of {\em Panor.
  Synth\`eses}, pages 1--119. Soc. Math. France, Paris, 2008.
\newblock With an appendix by Fr\'{e}d\'{e}ric Klopp.

\bibitem{KM2007}
Werner Kirsch and Bernd Metzger.
\newblock The integrated density of states for random {S}chr\"{o}dinger
  operators.
\newblock In {\em Spectral theory and mathematical physics: a {F}estschrift in
  honor of {B}arry {S}imon's 60th birthday}, volume~76 of {\em Proc. Sympos.
  Pure Math.}, pages 649--696. Amer. Math. Soc., Providence, RI, 2007.

\bibitem{kirschnitzscher}
Werner Kirsch and Frank Nitzschner.
\newblock Lifshitz-tails and non-{L}ifshitz-tails for one-dimensional random
  point interactions.
\newblock In {\em Order, disorder and chaos in quantum systems ({D}ubna,
  1989)}, volume~46 of {\em Oper. Theory Adv. Appl.}, pages 171--178.
  Birkh\"{a}user, Basel, 1990.

\bibitem{kloppnakamura09}
Fr\'{e}d\'{e}ric Klopp and Shu Nakamura.
\newblock Spectral extrema and {L}ifshitz tails for non-monotonous alloy type
  models.
\newblock {\em Comm. Math. Phys.}, 287(3):1133--1143, 2009.

\bibitem{Klopp2010}
Fr{\'e}d{\'e}ric Klopp and Shu Nakamura.
\newblock {L}ifshitz tails for generalized alloy-type random {S}chr{\"o}dinger
  operators.
\newblock {\em Analysis \& PDE}, 3(4):409--426, 2010.

\bibitem{Leschke-Warzel-2004}
Hajo Leschke and Simone Warzel.
\newblock Quantum-classical transitions in lifshitz tails with magnetic fields.
\newblock {\em Phys. Rev. Lett.}, 92:086402, Feb 2004.

\bibitem{Letac2017}
G\'{e}rard Letac and Jacek Weso{\l}owski.
\newblock Multivariate reciprocal inverse {G}aussian distributions from the
  {S}abot-{T}arr\`es-{Z}eng integral.
\newblock {\em J. Multivariate Anal.}, 175:104559, 18, 2020.

\bibitem{LP}
Russell Lyons and Yuval Peres.
\newblock {\em Probability on trees and networks}, volume~42 of {\em Camb. Ser.
  Stat. Probab. Math.}
\newblock Cambridge: Cambridge University Press, 2016.

\bibitem{mulsto}
Peter M\"{u}ller and Peter Stollmann.
\newblock Spectral asymptotics of the {L}aplacian on supercritical
  bond-percolation graphs.
\newblock {\em J. Funct. Anal.}, 252(1):233--246, 2007.

\bibitem{Muller2011}
Peter M\"{u}ller and Peter Stollmann.
\newblock Percolation {H}amiltonians.
\newblock In {\em Random walks, boundaries and spectra}, volume~64 of {\em
  Progr. Probab.}, pages 235--258. Birkh\"{a}user/Springer Basel AG, Basel,
  2011.

\bibitem{najar09}
Hatem Najar.
\newblock Non-{L}ifshitz tails at the spectral bottom of some random operators.
\newblock {\em J. Stat. Phys.}, 130(4):713--725, 2008.

\bibitem{Poudevigne2019}
R\'emy Poudevigne.
\newblock Monotonicity and phase transition for the {VRJP} and the {ERRW}.
\newblock {\em arXiv:1911.02181 [math]}, 2019.

\bibitem{theserapenne}
Valentin Rapenne.
\newblock {\em Etude de quelques probl{\`e}mes li{\'e}s aux marches
  al{\'e}atoires avec renforcement et aux marches branchantes}.
\newblock PhD thesis.
\newblock To be published.

\bibitem{rojas2013anderson}
Constanza Rojas-Molina.
\newblock The {A}nderson model with missing sites.
\newblock {\em Oper. Matrices}, 8(1):287--299, 2014.

\bibitem{Sabot2015}
Christophe Sabot and Pierre Tarr\`es.
\newblock {{Edge-reinforced random walk, vertex-reinforced jump process and the
  supersymmetric hyperbolic sigma model}}.
\newblock {\em J. Eur. Math. Soc.}, 17(9):2353--2378, 2015.

\bibitem{Sabot2017}
Christophe Sabot, Pierre Tarr\`es, and Xiaolin Zeng.
\newblock The vertex reinforced jump process and a random {S}chr\"{o}dinger
  operator on finite graphs.
\newblock {\em Ann. Probab.}, 45(6A):3967--3986, 2017.

\bibitem{Sabot2019}
Christophe Sabot and Xiaolin Zeng.
\newblock A random {S}chr\"{o}dinger operator associated with the vertex
  reinforced jump process on infinite graphs.
\newblock {\em J. Amer. Math. Soc.}, 32(2):311--349, 2019.

\bibitem{Zirnbauer1991}
Martin~R. Zirnbauer.
\newblock {F}ourier analysis on a hyperbolic supermanifold with constant
  curvature.
\newblock {\em Comm. Math. Phys.}, 141(3):503--522, 1991.

\end{thebibliography}
\end{document}